\documentclass{amsart}
\usepackage{amsmath,amsthm}
\usepackage{amsfonts,amssymb}

\usepackage{rotating}
\usepackage{graphicx}
\usepackage{longtable}
\usepackage{tikz}

\usepackage{enumerate}

\newtheorem{theorem}{Theorem}[section]
\newtheorem{corollary}[theorem]{Corollary}
\newtheorem{lemma}[theorem]{Lemma}
\newtheorem{proposition}[theorem]{Proposition}

\newtheorem{definition}[theorem]{Definition}
\newtheorem{remark}[theorem]{Remark}

\theoremstyle{remark}

\newcommand{\mR}{\mathbb{R}}
\newcommand{\mC}{\mathbb{C}}
\newcommand{\mN}{\mathbb{N}}
\newcommand{\mE}{\mathbb{E}}

\newcommand{\mS}{\mathbb{S}}

\newcommand{\mD}{\mathbb{D}}

\newcommand{\cH}{\mathcal{H}}

\newcommand{\cF}{\mathcal{F}}
\newcommand{\cP}{\mathcal{P}}
\newcommand{\cC}{\mathcal{C}}

\newcommand{\cS}{\mathcal{S}}

\newcommand{\cV}{\mathcal{V}}
\newcommand{\cO}{\mathcal{O}}

\newcommand{\ux}{\underline{x}}
\newcommand{\uxb}{\underline{x} \grave{}}

\newcommand{\uyb}{\underline{y} \grave{}}
\newcommand{\uy}{\underline{y}}

\newcommand{\pI}{\partial_{x_i}}

\newcommand{\upx}{\partial_{\underline{x}}}

\begin{document}
 
\title[Hilbert space for quantum mechanics on superspace]
{Hilbert space for quantum mechanics on superspace}

\author{Kevin Coulembier}
\address{Department of Mathematical Analysis\\
Ghent University\\ Krijgslaan 281, 9000 Gent\\ Belgium.}
\email{Coulembier@cage.ugent.be}
\author{Hendrik De Bie}
\address{Department of Mathematical Analysis\\
Ghent University\\ Krijgslaan 281, 9000 Gent\\ Belgium.}
\email{Hendrik.DeBie@UGent.be}

\date{\today}
\thanks{K. Coulembier is as Ph.D. Fellow of the Research Foundation - Flanders (FWO). H. De Bie is a Postdoctoral Fellow of the Research Foundation - Flanders (FWO)}

\maketitle

\begin{abstract}
In superspace a realization of $\mathfrak{sl}_2$ is generated by the super Laplace operator and the generalized norm squared. In this paper, an inner product on superspace for which this representation is skew-symmetric is considered. This inner product was already defined for spaces of weighted polynomials (see [K. Coulembier, H. De Bie and F. Sommen, Orthogonality of Hermite polynomials in superspace and Mehler type formulae, arXiv:1002.1118]). In this article, it is proven that this inner product can be extended to the super Schwartz space, but not to the space of square integrable functions. Subsequently, the correct Hilbert space corresponding to this inner product is defined and studied. A complete basis of eigenfunctions for general orthosymplectically invariant quantum problems is constructed for this Hilbert space. Then the integrability of the $\mathfrak{sl}_2$-representation is proven. Finally the Heisenberg uncertainty principle for the super Fourier transform is constructed.
\end{abstract}



\section{Introduction}

In recent work, we have been developing a new approach to the study of superspace, namely by means of harmonic analysis (see e.g. \cite{DBS5, DBE1, Cauchy, DBS9}). We work over a superspace $\mR^{m|2n}$ generated by $m$ commuting or bosonic variables and $2n$ anti-commuting or fermionic variables (\cite{MR732126}). The main feature of this approach is the introduction of an orthosymplectic super Laplace operator $\nabla^2$ and a generalized norm squared $R^2$. They have the property $\nabla^2 (R^2)=2(m-2n)=2M$ with $M$ the so-called super-dimension. This parameter characterizes several global features of the superspace $\mR^{m|2n}$, see \cite{DBS5, DBE1, DBS3, CDBS3}. In \cite{DBS5, DBE1, Cauchy} integration over the supersphere (algebraically defined by $R^2=1$) was introduced giving a new tool in the study of super analysis.

Schr\"odinger equations in superspace were considered first as a method to incorporate spin, see e.g. \cite{MR1292819, MR1019514}. The quantum (an-)harmonic oscillator (\cite{DBS3, MR1019514, MR967935}), the Kepler problem (\cite{MR2395482}), the delta potential (\cite{DBS8}) and the $CMS$-model (\cite{MR2025382}) have already been studied in superspace. Until now, a suitable Hilbert space structure for quantum mechanics on superspace has not been introduced. It is therefore the main aim of this paper to tackle that problem.  

In \cite{CDBS2}, we already constructed an inner product $\langle \cdot|\cdot\rangle_2$ on the space of polynomials weighted with the super Gaussian such that $R^2$ and $\nabla^2$ are symmetric operators. Here, we will show that this inner product can be extended to the super Schwartz space, $\cS(\mR^m)\otimes \Lambda_{2n}$, with $\Lambda_{2n}$ the Grassmann algebra, but not to $L_2(\mR^m)\otimes\Lambda_{2n}$. In order to prove the extension to $\cS(\mR^m)\otimes \Lambda_{2n}$ we construct a spherical Hermite representation theorem for $\cS(\mR^m)\otimes\Lambda_{2n}$. In doing so, we will consider two different bases of the Schwartz space, the product basis (obtained by products of purely bosonic and fermionic bases) and the spherical basis (generated by the canonical $\mathfrak{sl}_{2}$ realization in superspace). The main technical tool in extending the inner product $\langle \cdot|\cdot\rangle_2$ is the determination of polynomial bounds of the change-of-basis coefficients between these two bases (see theorem \ref{afschatting1}). 

Next we turn our attention to suitable Hilbert spaces for quantum mechanics in superspace. We first prove the rather surprising fact that any Hilbert space for which $R^2$ and $\nabla^2$ are symmetric operators and which contains the eigenvectors of the harmonic oscillator will contain generalized functions (i.e. distributions). These generalized functions are weak derivatives of functions. This is in agreement with the study of orthosymplectically invariant Schr\"odinger equations in \cite{CDBS3} which led to the conclusion that it is natural to include generalized functions in the space of solutions. Subsequently, we define our `canonical' Hilbert space as the closure of $\cS(\mR^m)\otimes \Lambda_{2n}$ with respect to $\langle \cdot|\cdot\rangle_2$. We show that the space of solutions of orthosymplectically invariant Schr\"odinger equations (such as the oscillator and Kepler problem) constructed in \cite{CDBS3} forms a complete basis. We also derive a criterion for essential self-adjointness of orthosymplectically invariant Hamiltonians. Then we show that the $\mathfrak{sl}_2$ representation generated by the super Laplace operator and the generalized norm squared on $\cS(\mR^m)\otimes\Lambda_{2n}$ is integrable. This is a generalization of a result for Dunkl-harmonic analysis, see \cite{MR2352481}. Finally we formulate and prove the Heisenberg uncertainty relation for the super Fourier transform.

The paper is organized as follows. First the necessary notions of super harmonic analysis are reviewed. Then the coefficients for the change of basis between the super Hermite functions and the product of purely bosonic and fermionic Hermite functions are calculated and a detailed study is made of their growth properties. Using this result we prove that the inner product can be continuously extended to $\cS(\mR^m)\otimes\Lambda_{2n}$, but not to $L_2(\mR^m)\otimes\Lambda_{2n}$. The Hilbert space $\bold{L}_2(\mR^{m|2n})$ is then constructed as the unique Hilbert space containing the space of polynomials weighted with the super Gaussian, corresponding to the inner product $\langle\cdot|\cdot\rangle_2$. Then we consider orthosymplectically invariant Schr\"odinger equations on superspace. Finally the integrability of the $\mathfrak{sl}_2$ representation and the uncertainty relation for the Fourier transform is proven. Also a list of notations is provided in order to avoid confusion.

\section{Harmonic analysis in superspace}
\setcounter{equation}{0}
\label{preliminaries}
Superspaces are spaces where one considers not only commuting (bosonic) but also anti-commuting (fermionic) co-ordinates (see a.o. \cite{MR732126}). The $2n$ anti-commuting variables ${x\grave{}}_i$ generate the complex Grassmann algebra $\Lambda_{2n}$ under the relations ${x\grave{}}_i{x\grave{}}_j=-{x\grave{}}_j{x\grave{}}_i$. An arbitrary element $f \in \Lambda_{2n}$ can hence always be written as $f = \sum_A f_A {x \grave{}}_A$ with ${x \grave{}}_A = {x \grave{}}_1^{\, \alpha_1} \ldots {x \grave{}}_{2n}^{\, \alpha_{2n}}$, $A = (\alpha_{1}, \ldots, \alpha_{2n}) \in \{0,1\}^{2n}$ and $f_A \in \mC$. The dimension of $\Lambda_{2n}$ as a $\mC$-vectorspace is hence $2^{2n}$. We consider a space with $m$ bosonic variables $x_i$. The supervector $\bold{x}$ is defined as
\[
\bold{x}=(X_1,\cdots,X_{m+2n})=(\ux,\uxb)=(x_1,\cdots,x_m,{x\grave{}}_1,\cdots,{x\grave{}}_{2n}).
\]

The commutation relations for the Grassmann algebra and the bosonic variables are captured in the relation $X_iX_j=(-1)^{[i][j]}X_jX_i$ with $[i]=0$ if $i\le m$ and $[i]=1$ otherwise. The super-dimension is defined as $M=m-2n$. We consider a Riemannian superspace $\mR^{m|2n}$ with the orthosymplectic metric $g$ defined as
\begin{equation*}
\begin{cases} 
g^{ii}=1 & 1\le i\le m,\\
g^{2i-1+m,2i+m}=-1/2 & 1\le i\le n ,\\  
g^{2i+m,2i-1+m}=1/2 & 1\le i\le n,\\
g^{ij}=0 & otherwise. 
\end{cases}
\end{equation*}
We define $X^j=X_ig^{ij}$. The square of the `radial coordinate' is given by
\begin{equation*}
R^2=\langle \bold{x},\bold{x}\rangle=\sum_{j=1}^{m+2n}X^jX_j =\sum_{i=1}^mx_i^2-\sum_{j=1}^n{x\grave{}}_{2j-1}{x\grave{}}_{2j}=r^2+\theta^2,
\end{equation*}
see \cite{CDBS3}. The super gradient is defined by
\begin{eqnarray*}
\nabla&=&(\partial_{X^1},\cdots,\partial_{X^{m+2n}})\\
&=&(\partial_{x_1},\cdots,\partial_{x_m},2\partial_{{x\grave{}}_{2}},-2\partial_{{x\grave{}}_{1}},\cdots,2\partial_{{x\grave{}}_{2n}},-2\partial_{{x\grave{}}_{2n-1}}).
\end{eqnarray*}
From this expression and the metric we obtain the Laplace operator
\begin{equation*}
\nabla^2=\langle \nabla,\nabla \rangle=\sum_{k=1}^{m+2n}\partial_{X_k}\partial_{X^k} =\sum_{i=1}^m\partial_{x_i}^2-4\sum_{j=1}^n\partial_{{x\grave{}}_{2j-1}}\partial_{{x\grave{}}_{2j}}.
\end{equation*}

The super Euler operator is defined by
\begin{equation*}
\mE=\mE_b+\mE_f=\langle x,\nabla\rangle=\sum_{k=1}^{m+2n}X^k\partial_{X^k}=\sum_{i=1}^mx_i\partial_{x_i}+\sum_{j=1}^{2n}{x\grave{}}_j\partial_{{x\grave{}}_j}.
\end{equation*}
The operators $i\nabla^2/2$, $iR^2/2$ and $\mE+M/2$, generate the $\mathfrak{sl}_2$ Lie-algebra (\cite{DBS5}). In particular, the relation
\begin{eqnarray}
\label{commRD}
\left[\nabla^2/2,R^2/2\right]&=&2\mE+M
\end{eqnarray} 
holds. The Laplace-Beltrami operator is defined as
\begin{equation}
\label{LB}
\Delta_{LB}=R^2\nabla^2-\mE(M-2+\mE).
\end{equation}

The inner product of two supervectors $\bold{x}$ and $\bold{y}$ is given by
\begin{equation}
\label{inprod}
\langle \bold{x},\bold{y}\rangle=\sum_{i=1}^mx_iy_i-\frac{1}{2}\sum_{j=1}^n({x\grave{}}_{2j-1}{y\grave{}}_{2j}-{x\grave{}}_{2j}{y\grave{}}_{2j-1})=\langle \ux,\uy\rangle+\langle\uxb,\uyb\rangle.
\end{equation}
The commutation relations for two supervectors are determined by the relation $X_iY_j=$ $(-1)^{[i][j]}Y_jX_i$. This implies that the inner product \eqref{inprod} is symmetric, i.e. $\langle \bold{x},\bold{y}\rangle=\langle\bold{y},\bold{x}\rangle$.

The space of super polynomials is given by $\cP=\mR[x_1,\cdots,x_m]\otimes \Lambda_{2n}$. The space of homogeneous polynomials of degree $k$ is denoted by $\cP_k$ and consists of the elements $P\in\cP$ which satisfy $\mE P=kP$. More general superfunctions can for instance be defined as functions with values in the Grassmann algebra, $f:\Omega\subset\mR^m\to\Lambda_{2n}$. They can always be expanded as  $f=\sum_A {x\grave{}}_A f_A$, with ${x\grave{}}_A$ the basis of monomials for the Grassmann algebra and $f_A:\Omega\subset\mR^m\to\mC$. In general, for a function space $\cF$ corresponding to the $m$ bosonic variables (e.g. $\cS(\mR^{m})$, $L_p(\mR^m)$, $C^k(\Omega)$) we use the notation $\cF_{m|2n}=\cF\otimes\Lambda_{2n}$.

The null-solutions of the super Laplace operator are called harmonic superfunctions. The space of all spherical harmonics of degree $k$ is denoted by $\cH_k=\cP_k\cap\ker \nabla^2$. Equation \eqref{LB} implies they are eigenfunctions of the Laplace-Beltrami operator
\begin{equation}
\label{LBH}
\Delta_{LB}\cH_k=-k(k+M-2)\cH_k.
\end{equation}
In the purely bosonic case we denote $\cH_{k}$ by $\cH_{k}^{b}$, in the purely fermionic case by $\cH_{k}^{f}$. We have the following decomposition (see \cite{DBS5}).

\begin{lemma}[Fischer decomposition]
If $M \not \in -2 \mN$, the space $\cP$ decomposes as $\cP =\bigoplus_{k=0}^{\infty}  \bigoplus_{j=0}^{\infty} R^{2j}\cH_k$. If $m=0$, then the decomposition is given by $\Lambda_{2n} = \bigoplus_{k=0}^{n} \bigoplus_{j=0}^{n-k} \theta^{2j} \cH^f_k$.
\label{scalFischer}
\end{lemma}

In \cite{DBE1} the vector space $\cH_k$ was decomposed into irreducible pieces under the action of $SO(m)\times Sp(2n)$. 
\begin{theorem}[Decomposition of $\cH_k$]
Under the action of $SO(m) \times Sp(2n)$ the space $\cH_k$ decomposes as
\label{decompintoirreps}
\label{polythm}
\[
\cH_{k} = \bigoplus_{j=0}^{\min(n, k)} \bigoplus_{l=0}^{\min(n-j,\lfloor \frac{k-j}{2} \rfloor)} f_{l,k-2l-j,j} \cH^b_{k-2l-j} \otimes \cH^f_{j}.
\]
The polynomials $f_{l,k-2l-j,j}$ are given by the formula 
\[
 f_{k,p,q}=\sum_{s=0}^ka_sr^{2k-2s}\theta^{2s} \;\;\;,\;\;\;\;\; a_s=\binom{k}{s}\frac{(n-q-s)!}{\Gamma (\frac{m}{2}+p+k-s)}\frac{\Gamma(\frac{m}{2}+p+k)}{(n-q-k)!}
\]
for $0\le q \le n$ and $0\le k\le n-q$. They are, up to normalization, the unique polynomials of degree $2k$ such that $f_{k,p,q} \cH_p^b \otimes \cH_q^f \neq 0$ and $\Delta (f_{k,p,q} \cH_p^b \otimes \cH_q^f) = 0$.
\end{theorem}
In particular $f_{0,p,q}=1$ holds.

The integration used on $\Lambda_{2n}$ is the so-called Berezin integral (see \cite{MR732126,DBS5}), defined by
\begin{equation}
\label{Berezin}
\int_{B} =\pi^{-n} \partial_{{x \grave{}}_{2n}} \ldots \partial_{{x \grave{}}_{1}} = \frac{ \pi^{-n}}{4^n n!} \nabla_f^{2n}.
\end{equation}
On a general superspace, the integration is then defined by
\begin{equation}
\label{superint}
\int_{\mR^{m | 2n}} = \int_{\mR^m} dV(\ux)\int_B=\int_B \int_{\mR^m} dV(\ux),
\end{equation}
with $dV(\ux)$ the usual Lebesgue measure in $\mR^{m}$.

The supersphere is algebraically defined by the relation $R^2=1$, for $m\not=0$. The integration over the supersphere was introduced in \cite{DBS5} for polynomials and generalized to a broader class of functions in \cite{Cauchy}. The integration for polynomials is uniquely defined (see \cite{DBE1, Cauchy}) by the following properties.

\begin{theorem}
\label{Pizzettithm}
For $m\not=0$, the only (up to a multiplicative constant) linear functional $T: \cP \rightarrow \mR$ satisfying the following properties for all $f(\bold{x}) \in \cP$:
\begin{itemize}
\item $T(R^2 f(\bold{x})) = T(f(\bold{x}))$
\item $T(f(A \cdot \bold{x})) = T(f(\bold{x}))$, \quad $\forall A \in SO(m)\times Sp(2n)$
\item $k \neq l \quad \Longrightarrow \quad T(\cH_k \cH_l) = 0$
\end{itemize}
is given by the Pizzetti integral
\begin{equation}
\label{Pizzetti}
\int_{SS} P  =  \sum_{k=0}^{\infty}  \frac{2 \pi^{M/2}}{2^{2k} k!\Gamma(k+M/2)} (\nabla^{2k} P )(0),\qquad P\in\cP.
\end{equation}
\end{theorem}

The orthogonality condition on the supersphere can be made even stronger.

\begin{theorem}
\label{integorth}
One has that 
\[
f_{i,p,q} \cH^b_{p} \otimes \cH^f_{q} \quad  \bot \quad f_{j,r,s} \cH^b_{r} \otimes \cH^f_{s}, \quad \mbox{i.e.} \int_{SS} \left(f_{i,p,q} \cH^b_{p} \otimes \cH^f_{q}\right)  \left(f_{j,r,s} \cH^b_{r} \otimes \cH^f_{s}\right) =0
\]
with respect to the Pizzetti integral if and only if $(i,p,q) \neq (j,r,s)$.
\end{theorem}

This supersphere integration is a dimensional continuation of the bosonic Pizzetti formula. The Berezin integral can be connected with the supersphere integration, again by dimensional continuation, see \cite{Cauchy}. For $M>0$ and $f$ a function in $L_1(\mR^m)_{m|2n}\cap C^n(\mR^m)_{m|2n}$, the following relation holds
\begin{equation}
\label{superint}
\int_{\mR^{m|2n}}f=\int_{0}^\infty dv\, v^{M-1}\int_{SS,x}f(vx).
\end{equation}
In particular for $P_k$ a polynomial of degree $k$, one has
\begin{equation}
\label{superint2}
\int_{\mR^{m|2n}}P_k\exp (-R^2) =\frac{1}{2}\Gamma (\frac{k+M}{2})\int_{SS} P_k.
\end{equation}

We repeat the notions of symmetric, essentially self-adjoint and self-adjoint operators on Hilbert spaces (see \cite{MR0751959}).
\begin{definition}
Let $\cO$ be a densely defined operator on a Hilbert space $\cV$ with adjoint $\mathcal{O}^\dagger$. $\mathcal{O}$ is called symmetric (or hermitian) if for the domains $\mD(\mathcal{O})\subset\mD(\mathcal{O}^\dagger)$ and if $\mathcal{O}\phi=\mathcal{O}^\dagger\phi$ for all $\phi\in\mD(\mathcal{O})$. $\mathcal{O}$ is self-adjoint if it is symmetric and if $\mD(\mathcal{O})=\mD(\mathcal{O}^\dagger)$. $\mathcal{O}$ is essentially self-adjoint if it is symmetric and its closure $\overline{\mathcal{O}}$ is self-adjoint.
\end{definition}
There is an important criterion for essential self-adjointness (\cite{MR1349825}).
\begin{lemma}
\label{critess}
If for a symmetric operator $\mathcal{O}$ on a Hilbert space $\cV$ with domain $\mD(\mathcal{O})$ there exists a complete orthonormal set $f_j$ for $\cV$ in $\mD(\mathcal{O})$ for which $\mathcal{O}f_j=\lambda_j f_j$, then $\mathcal{O}$ is essentially self-adjoint.
\end{lemma}
For finite dimensional vector spaces the notions of symmetric and self-adjoint operator coincide.

The orthosymplectic Lie superalgebra $\mathfrak{osp}(m|2n)$ is generated by the following differential operators (see \cite{CDBS3, MR2395482})
\begin{equation}
\label{ospgen}
L_{ij}=X_i\partial_{X^j}-(-1)^{[i][j]}X_j\partial_{X^i},
\end{equation}
for $1\le i\le j\le m+2n$.

Schr\"odinger equations in superspace are equations of the type
\[
-\frac{\nabla^2}{2} \psi + V(\bold{x}) \psi = E \psi
\]
with wave function $\psi \in L_{2}(\mR^{m})_{m|2n}$, the potential $V$ a superfunction and the energy $E$ a complex number. Several authors have studied such equations. The (purely fermionic) harmonic oscillator was studied in \cite{MR830398}. Anharmonic extensions were studied in \cite{MR967935, MR1019514, MR1032208}. In \cite{MR2395482}, Zhang studied the hydrogen atom in superspace (or quantum Kepler problem) using Lie superalgebra techniques. Also the delta potential has been studied, see \cite{DBS8}. 

In this paper we will consider orthosymplectically invariant Schr\"odinger equations, which means $L_{ij}V=0$ for all $i,j$. Theorem 3 in \cite{CDBS3} implies that such potentials are of the form $V(R^{2})$. Here, functions of the type $V(R^2)$ are defined as follows:
\begin{definition}
\label{sphsymm}
For a function $h:\mR^+\to\mR$, $h\in C^n(\mR^+)$, the superfunction with notation $h(R^2)$ is defined as $h(R^2)=\sum_{j=0}^n\frac{\theta^{2j}}{j!}h^{(j)}(r^2)$ and is an element of $C^{0}(\mR^{m})_{m|2n}$. 
\end{definition}

When $h$ is sufficiently smooth one has
\begin{equation}
\label{laplradharm}
\nabla^2 h(R^2)H_k=4R^2h^{(2)}(R^2)H_k+(4k+2M)h^{(1)}(R^2)H_k, \quad H_{k} \in \cH_{k}
\end{equation}
and
\begin{equation}
\label{LBR}
\left[\Delta_{LB},h(R^2)\right]=0.
\end{equation}

The simplest orthosymplectically invariant quantum problem is the harmonic oscillator described by the Hamiltonian
\begin{equation}
\label{hamiltoniaan}
H= \frac{1}{2}(R^2-\nabla^2) = \sum_{i=1}^m a_i^+ a_i^- +\sum_{i=1}^{2n} b_i^+ b_i^- + \frac{M}{2}
\end{equation}
with
\[
\begin{array}{llll}
a^+_i = \frac{\sqrt{2}}{2}(x_i - \pI ) &\quad&a^-_i = \frac{\sqrt{2}}{2}(x_i+\pI)\\
\vspace{-1mm}\\
b^+_{2i} = \frac{1}{2}({x \grave{}}_{2i} + 2 \partial_{{x \grave{}}_{2i-1}})&\quad&b^-_{2i} =\frac{1}{2} ({x \grave{}}_{2i-1} + 2 \partial_{{x \grave{}}_{2i}})\\
\vspace{-1mm}\\
b^+_{2i-1} = \frac{1}{2}({x \grave{}}_{2i-1} - 2 \partial_{{x \grave{}}_{2i}})&\quad&b^-_{2i-1} = \frac{1}{2}(-{x \grave{}}_{2i} + 2 \partial_{{x \grave{}}_{2i-1}})\\
\vspace{-1mm}\\
\end{array}
\]
the bosonic and fermionic creation and annihilation operators. When $M\not\in-2\mN$, the so-called spherical Hermite functions constitute a basis of eigenvectors of $H$ for $\cP\exp(-R^2/2)$, see the subsequent lemma \ref{eigCH}. In the purely bosonic case they give an alternative for the basis of cartesian Hermite functions, based on the $O(m)$-invariance of the Hamiltonian of the harmonic oscillator. In superspace they are defined in \cite{DBS3}. For an overview of the different types of Hermite functions and their properties, see \cite{CDBS2}. When $M\not\in-2\mN$, they can be expressed as
\begin{equation}
\label{CH}
\phi_{j,k,l}(\bold{x})=\sqrt{\frac{2j!}{\Gamma(j+k+\frac{M}{2})}}L_j^{\frac{M}{2}+k-1}(R^2)H_k^{(l)}(\bold{x})\exp(-R^2/2),
\end{equation}
for $j,k\in\mN$ and $l=1,\cdots,\dim\cH_k$, with $H_k^{(l)}$ a basis of $\cH_k$ and $L_j^\alpha$ the generalized Laguerre polynomials which  are defined as
\[
L_j^{\alpha}(t)=\sum_{k=0}^{j}\frac{(k+\alpha+1)_{j-k}}{k!(j-k)!}(-t)^k,
\]
with $(a)_j=(a)(a+1)\cdots(a+j-1)$ the Pochhammer symbol. We denote $\zeta_{j,k}^M=\sqrt{\frac{\Gamma(j+k+\frac{M}{2})}{2j!}}$. 

When $M>0$, the case with $M$ bosonic variables is a special case. We use the symbol $\uy$ for the vector variable in $\mR^M$. The definition of Hermite functions then yields
\begin{equation}
\label{CHbosM}
\phi_{j,k,l}(\uy)=\sqrt{\frac{2j!}{\Gamma(j+k+\frac{M}{2})}}L_j^{\frac{M}{2}+k-1}(r_{\uy}^2)H_{M,k}^{(l)}(\uy)\exp(-r^2_{\uy}/2),
\end{equation}
with $H_{M,k}^{(l)}$ a bosonic spherical harmonic in $M$ dimensions. For this case the space of spherical harmonics of degree $k$ is denoted $\cH_{k,M}^b$. 

In the purely bosonic case with $m$ dimensions, we will denote the spherical Hermite functions by $\phi^b_{j,k,l}(\ux)$,
\begin{equation}
\label{CHbosm}
\phi^b_{j,k,l}(\ux)=\sqrt{\frac{2j!}{\Gamma(j+k+\frac{m}{2})}}L_j^{\frac{m}{2}+k-1}(r^2)H_k^{b(l)}(\ux)\exp(-r^2/2),
\end{equation}
with $H_k^{b(l)}(\ux)\in\cH_k^b$. The distinction between the functions \eqref{CHbosM} and \eqref{CHbosm} is clearly artificial in some sense since $m$ and $M$ are both variables. The distinction will however be useful in the sequel as we will consider undetermined but fixed dimensions $m$ and $2n$ (leading to a fixed $M=m-2n$).

Since we will only use the spherical Hermite functions in this article we will simply call them the Hermite functions. For the sequel the following relations will be necessary, for proofs see \cite{DBS3}.
\begin{lemma}
\label{eigCH}
Let $M>0$. Then the super Hermite functions satisfy
\begin{eqnarray*}
\left(\nabla^2+R^2-2\mE-M)\right)\phi_{j,k,l}&=&-4\sqrt{(j+1)(j+k+\frac{M}{2})}\phi_{j+1,k,l}\\
\left(\nabla^2+R^2+2\mE+M)\right)\phi_{j,k,l}&=&-4\sqrt{j(j+k+\frac{M}{2}-1)}\phi_{j-1,k,l}\\
\frac{1}{2}\left(R^2-\nabla^2\right)\phi_{j,k,l}&=&(2j+k+\frac{M}{2})\phi_{j,k,l}.
\end{eqnarray*}
\end{lemma}

The fermionic Gaussian function is given by the finite Taylor expansion $\exp(-\theta^2/2) =\sum_{j=0}^n(-1)^j\frac{\theta^{2j}}{2^jj!}$. If we consider a basis $H_q^{f(t)}$ of $\cH^f_q$, then the functions
\begin{equation}
\phi_{s,q,t}^{f}(\uxb) = \sqrt{s!(n-s-q)!}L_s^{q-n-1}(\theta^2)H_q^{f(t)}\exp(-\theta^2/2)
\label{CliffordHermiteFunctions}
\end{equation}
with $s = 0, \ldots, n$; $q = 0, \ldots, n-s$ and $t = 1, \ldots, \dim \cH_k^{f}$, are the fermionic Hermite functions. They constitute a basis of $\Lambda_{2n}$, see lemma \ref{scalFischer}. We denote $\zeta_{s,q}^f=1/\sqrt{s!(n-s-q)!}$. The analogue of lemma \ref{eigCH} can be found in \cite{DBS3}. The following formula is important for the sequel,
\begin{equation}
\label{annferm}
(\nabla^2_f+\theta^2+2\mE_f-2n)\phi^f_{s,q,t}=4\sqrt{s(n-s-q+1)}\phi^f_{s-1,q,t}.
\end{equation}

The bosonic Hermite functions \eqref{CHbosm} are orthogonal with respect to the $L_2(\mR^m)$-inner product. The inner product for the fermionic Hermite functions is defined using the Hodge star map, see \cite{CDBS2, MR2025382, MR1796030}.
\begin{definition}
\label{defstar}
The star map $\ast$ maps monomials ${x \grave{}}_A = {x \grave{}}_1^{\, \alpha_1} \ldots {x \grave{}}_{2n}^{\, \alpha_{2n}}$ of degree $k$ to monomials $\ast {x \grave{}}_A = \pm 2^{k-n} {x \grave{}}_1^{1-\alpha_1} \ldots {x \grave{}}_{2n}^{1-\alpha_{2n}}$ of degree $(2n-k)$ where the sign is chosen such that ${x \grave{}}_A (\ast {x \grave{}}_A )= 2^{k-n}  {x \grave{}}_{1} \ldots {x \grave{}}_{2n}$. By linearity, $\ast$ is extended to the whole of $\Lambda_{2n}$.
\end{definition}

For $H_q^{f(t)}\in \cH_q^f$ and $s+q\le n$, the following property is proven in \cite{CDBS2},
\begin{equation}
\label{CHstar}
*L_s^{q-n-1}(\theta^2) H_q^{f(t)}\exp(-\theta^2/2)=(-1)^s L_s^{q-n-1}(\theta^2) \widetilde{H}_q^{f(t)}\exp(-\theta^2/2).
\end{equation}
Here, the transformation $\widetilde{.}:\Lambda_{2n}\to \Lambda_{2n} $ is a linear transformation defined by
\[
\widetilde{{x \grave{}}_{2i-1}}={x \grave{}}_{2i},\qquad \widetilde{{x \grave{}}_{2i}}=-{x \grave{}}_{2i-1}\quad\mbox{and}\qquad \widetilde{a\,b}=\widetilde{b}\,\widetilde{a} \quad \mbox{for}\quad a,b \in \Lambda_{2n}.
\]
If $H_q^f\in\cH_q^f$, then $\widetilde{H^f_q}\in\cH_q^f$. Since clearly $\widetilde{\widetilde{H_q^f}}=(-1)^qH_q^f$ we can always consider a basis $\{H_q^{f(t)}\}$ ($t=1, \ldots, \dim \cH_{q}^{f}$) of $\cH_q^f$ such that $\widetilde{H_q^{f(t)}}=\pm i^qH_q^{f(t)}$.

\begin{definition}
\label{inprodGrass}
The inner product $\langle\cdot|\cdot\rangle_{\Lambda_{2n}}:\Lambda_{2n}\times \Lambda_{2n}\to \mC$ is given by
\[
\langle f|g \rangle_{\Lambda_{2n}} = \int_{B,x} f (\ast \overline{g})= \frac{1}{(2\pi)^n} \sum_A 2^{|A|} f_A \overline{g_{A}},
\]
with $\overline{\, \cdot\, }$ the standard complex conjugation, $f=\sum_{A}f_A{x\grave{}}_A$ and  $g=\sum_{A}g_A{x\grave{}}_A$. 
\end{definition}
With respect to the inner product $\langle . | . \rangle_{\Lambda_{2n}} $ the adjoints of $\theta^2$, $\nabla^2_f$ and $\mE_f-n$ are given by
\begin{equation*}
(\theta^2)^\dagger=-\nabla^2_f, \qquad (\nabla^2_f)^\dagger=-\theta^2, \qquad (\mE_f-n)^\dagger=(\mE_f-n).
\end{equation*}
This implies that the fermionic harmonic oscillator is hermitian with respect to this inner product. However, other symplectically invariant Hamiltonians will not be. This is as expected, as in \cite{MR967935} it was calculated that the eigenvalues for the Hamiltonian $H=\nabla^2_f-\theta^2+\lambda \theta^4$ (with $\lambda$ real) can be complex. 

Now choose a basis $\{ H_q^{f(t)}\}$ of $\cH_q^f$ such that
\begin{equation}
\langle H_q^{f(t_1)} \exp(-\theta^2/2)|H_q^{f(t_2)} \exp(-\theta^2/2) \rangle_{\Lambda_{2n}} = \frac{\delta_{t_1t_2}}{(n-q)!}.
\label{OrthSphHarm}
\end{equation}
Using property \eqref{CHstar} we find that if $H_q^{f(t_1)}$ and $H_q^{f(t_2)}$ are eigenvectors of $\widetilde{\cdot}$ with different eigenvalues they are orthogonal, so we can still find an orthogonal basis for $\cH_q^f$ for which the elements satisfy $\widetilde{H_q}=\pm i^qH_q$. From now on we assume we use such a basis. Theorem $4.16$ in \cite{CDBS2} proves the orthonormality of the fermionic Hermite functions.

\begin{theorem}
\label{orthocliffherm}
The Hermite functions defined in equation (\ref{CliffordHermiteFunctions}) are orthonormal with respect to the inner product $\langle \cdot|\cdot\rangle_{\Lambda_{2n}}$:
\[
\langle \phi_{s,q,t}^{f} |\phi_{i,j,r}^{f} \rangle_{\Lambda_{2n}} = \delta_{si}\delta_{qj}\delta_{tr}.
\]
\end{theorem}

The natural inner product on superspace is given by the direct product of the $L_2(\mR^m)$-inner product and the inner product on the Grassmann algebra, see also \cite{CDBS2, MR2025382}.
\begin{definition}
\label{Canonischin}
The inner product $\langle \cdot|\cdot \rangle_1$, $L_2(\mR^{m})_{m|2n}\times L_2(\mR^{m})_{m|2n}\to \mC$ is given by
\[
\langle f | g\rangle_1 = \int_{\mR^{m|2n}} f (* \overline{g})=\int_{\mR^m} \langle f|g\rangle_{\Lambda_{2n}} dV(\ux)
\]
where the star map acts on $\Lambda_{2n}$ as in definition \ref{defstar} or equation \eqref{CHstar} and leaves the bosonic variables invariant.
\end{definition}
This inner product generates the topology of $L_2(\mR^m)_{m|2n}$. For this Hilbert space 
\begin{eqnarray}
\label{adj1}
(R^2)^\dagger=r^2-\nabla^2_f&\mbox{ and }&(\nabla^2)^\dagger=\nabla_b^2-\theta^2,
\end{eqnarray}
so the harmonic oscillator (\ref{hamiltoniaan}) is symmetric, but more general orthosymplectically invariant Hamiltonians will not be. Moreover, the super Hermite polynomials are not orthogonal with respect to this inner product, see example $5.3$ in \cite{CDBS2}.

In \cite{CDBS2} an inner product was constructed for which $R^2$ and $\nabla^2$ are symmetric. It was shown that this construction could only be made in case $M>0$, see theorem $5.15$ in \cite{CDBS2}. For the remainder of this article we hence always assume $M>0$, unless stated otherwise. The inner product $\langle\cdot|\cdot\rangle_2$ on $\cP\exp(-R^2/2)$, for which $R^2$ and $\nabla^2$ are symmetric, is defined using the linear map $T$: $\cP\exp(-R^2/2)\to\cP\exp(-R^2/2)$, defined by
\begin{equation}
\label{defT}
T\left[R^{2j}f_{k,p,q}H_p^bH_q^f\exp(-R^2/2)\right]=(-1)^{k}R^{2j}f_{k,p,q}H_p^b\widetilde{H}_q^f\exp(-R^2/2),
\end{equation}
with $f_{k,p,q}$ the polynomials determined in theorem \ref{polythm}.
 
\begin{theorem}
\label{defsuper}
When $M>0$, the product $\langle \cdot|\cdot \rangle_{2}$, $\cP \exp (-R^2/2)\times \cP \exp (-R^2/2)\to \mC$ given by
\[
\langle f | g\rangle_2 = \int_{\mR^{m|2n}} fT(\overline{g})
\]
is an inner product. One has $(R^2)^{\dagger}=R^2$ and $(\nabla^2)^{\dagger}=\nabla^2$.
\end{theorem}
\begin{proof}
See theorem $5.11$ and lemma $5.15$ in \cite{CDBS2}.
\end{proof}

The results of lemma $5.6$ and $5.9$ of \cite{CDBS2} are summarized in the following lemma.
\begin{lemma}
\label{SSin4}
\label{superorthbasis}
Take $\{H_p^{b(l)}\}$ ($l=1, \ldots, \dim \cH_{p}^{b}$) to be an orthonormal basis for $\cH_p^b$ satisfying
\begin{equation}
\label{bosharmbasis}
\int_{\mS^{m-1}} H_p^{(l_{1})}(\xi) \overline{H_p^{(l_{2})}}(\xi) d \sigma(\xi) = \delta_{l_{1} l_{2}}.
\end{equation}
Take $\{\cH_q^{f(t)}\}$ ($t=1, \ldots, \dim \cH_{q}^{f}$) the orthonormal basis of fermionic spherical harmonics of degree $q$ in equation \eqref{OrthSphHarm}. For $f_{k,p,q}$ as defined in theorem \ref{polythm} the following relation holds,
\begin{eqnarray*}
\int_{SS}f_{k,p,q}H_p^{b(l_1)}H_q^{f(t_1)}\, f_{k,p,q}H_p^{b(l_2)}\widetilde{H}_q^{f(t_2)}=(-1)^ka_{k,p,q}b_{k,p,q} \delta_{l_{1} l_{2}}\delta_{t_{1} t_{2}},
\end{eqnarray*}
with $a_{k,p,q}=\frac{\Gamma(M/2+p+q+2k-1)}{\Gamma(M/2+p+q+k-1)}$ and
\begin{eqnarray}
\label{bkpq}
b_{k,p,q}&=&\int_{SS}r^{2k} f_{k,p,q}H_p^{b(l)}H_p^{b(l)} H_q^{f(t)}\widetilde{H}_q^{f(t)}\\
\nonumber
&=&\frac{k!}{{\Gamma(2k+\frac{M}{2}+p+q)}}\frac{\Gamma(\frac{m}{2}+p+k)}{(n-q-k)!}.
\end{eqnarray}
Then, the basis for $\cH$ given by
\[
H_{2k+p+q}^{\left(r[k,p,q,l,t]\right)}=\frac{f_{k,p,q}H_p^{b(l)}H_q^{f(t)}}{\sqrt{a_{k,p,q}b_{k,p,q}}}
\]
with $k,p,q,l,t\in\mN$, $0\le k\le n$, $0\le q\le n-k$, $1\le l\le\dim\cH_p^b$ and $1\le t\le\dim\cH_q^f$ and $r[k,p,q,l,t]$, for $2k+p+q$ fixed, an injective map onto $\{1,\cdots,\dim\cH_{2k+p+q}\}$, satisfies the relation
\begin{equation}
\label{orthbasissuper}
\langle H_k^{(l)}\exp(-R^2/2) | H_k^{(r)}\exp(-R^2/2)\rangle_2 = \delta_{lr}\frac{1}{2}\Gamma(\frac{M}{2}+k).
\end{equation}
\end{lemma}

The super Hermite functions are then orthogonal with respect to the inner product $\langle \cdot|\cdot \rangle_{2}$, see theorem $5.13$ in \cite{CDBS2}.
\begin{theorem}
\label{orthCH}
The set of functions $\{ \phi_{j,k,l} \}$ in formula \eqref{CH} with the basis of spherical harmonics in formula \eqref{orthbasissuper} forms an orthonormal basis for $\cP\, exp(-R^{2}/2)$ with respect to the inner product $\langle \cdot|\cdot\rangle_2$,
\begin{equation}
\langle \phi_{j,k,l}| \phi_{i,s,r} \rangle_2 =\delta_{ji}\delta_{ks}\delta_{lr}.
\label{normCH}
\end{equation}
\end{theorem}

Finally, the super Fourier transform on $\cS(\mR^m)\otimes \Lambda_{2n}$ was introduced in \cite{DBS9} as
\begin{equation}
\label{Four}
\cF^{\pm}_{m|2n}(f(\bold{x}))(\bold{y})=(2\pi)^{-M/2}\int_{\mR^{m|2n},x}\exp(\pm i\langle x,y\rangle)f(\bold{x}).
\end{equation}
This yields an $O(m)\times Sp(2n)$-invariant generalization of the purely bosonic Fourier transform. The super Hermite functions are the eigenvectors of this Fourier transform, i.e.
\begin{equation}
\label{FourCH}
\cF^{\pm}_{m|2n}(\phi_{j,k,l}(\bold{x}))(\bold{y})=\exp(\pm i(2j+k)\frac{\pi}{2})\phi_{j,k,l}(\bold{y}).
\end{equation}
The Fourier transform can be immediately generalized to $L_2(\mR^m)_{m|2n}$ and $\cS'(\mR^m)\otimes\Lambda_{2n}$ (with  $\cS'(\mR^m)$ the space of tempered distributions) since $\cF^{\pm}_{m|2n}=\cF^{\pm}_{m|0}\circ\cF^{\pm}_{0|2n}$. An important property of the super Fourier transform is
\begin{eqnarray}
\label{FourR}
R^2\cF_{m|2n}^\pm(f)(\bold{x})&=&-\cF_{m|2n}^\pm(\nabla^2f)(\bold{x}).
\end{eqnarray}

\section{The super Hermite functions and the product basis.}
\setcounter{equation}{0}

In this technical section we calculate how the super Hermite functions can be expanded in the basis consisting of products of the bosonic and fermionic Hermite functions. We take $2$ fixed integers $m$ and $n$ satisfying $M=m-2n>0$. We consider the $\langle \cdot |\cdot \rangle_2$ orthonormal basis of super Hermite functions for $\cP\exp(-R^2/2)$ on $\mR^{m|2n}$, using the spherical harmonics in lemma \ref{superorthbasis}
\begin{equation}
\label{normeringL2}
\phi_{j,k,p,q,l,t}=\frac{L_{j}^{\frac{M}{2}+2k+p+q-1}(R^2)\,f_{k,p,q}H_p^{b(l)}H_q^{f(t)}\,\exp (-R^2/2)}{\zeta_{j,2k+p+q}^M\, \sqrt{a_{k,p,q}\, b_{k,p,q}}}.
\end{equation}

The Hermite functions on $\mR^m$ and in $\Lambda_{2n}$ are denoted by $\phi^b_{i,p,l}$ and $ \phi^f_{s,q,t}$. The product basis $\{\phi^b_{i,p,l}  \phi^f_{s,q,t} \}$ of $\cP\exp(-R^2/2)$ is orthonormal with respect to $\langle\cdot|\cdot\rangle_1$ since $ \{\phi^b_{i,p,l}\}$ is an orthonormal basis for the $L_{2}(\mR^{m})$-inner product and $\{\phi^f_{s,q,t} \}$ is an orthonormal basis for the $\Lambda_{2n}$-inner product (theorem \ref{orthocliffherm}). The coefficients of an $L_{2}(\mR^{m})_{m|2n}$ function $f$ (so in particular for an $\cS(\mR^m)_{m|2n}$ function) with respect to this basis are calculated by
\begin{equation*}
\langle \phi^b_{i,p,l}\phi^f_{s,q,t}|f\rangle_1=\langle \phi^b_{i,p,l}| \langle \phi^f_{s,q,t}|f\rangle_{\Lambda_{2n}}\rangle_{L_2(\mR^m)}.
\end{equation*}
Since the bosonic and fermionic Hermite functions are eigenvectors of the bosonic and fermionic harmonic oscillator, the following relation holds,
\begin{equation}
\label{HOprod}
\frac{1}{2}(R^2-\nabla^2)\phi_{i,p,l}^b\phi^f_{s,q,t}=(2i+p+2s+q+\frac{M}{2})\phi_{i,p,l}^b\phi_{s,q,t}^f.
\end{equation}

We start to calculate the coefficients corresponding to the change of basis between the product basis and the super Hermite functions. Most of these coefficients are zero.
\begin{lemma}
\label{basisovergangsbf}
For the product basis $\{ \phi^b_{i,p,l} \phi^f_{s,q,t} \}$, with $i,p,s,q\in\mN$ and $s+q\le n$ and the basis in equation (\ref{normeringL2}), $\{\phi_{j,k,p,q,l,t}\}$ with $j,k\in\mN$ and $k\le n-q$, the following relation holds
\begin{equation}
\label{alphavgl}
\phi_{j,k,p,q,l,t}=\sum_{s=0}^{\min(n-q,j+k)}\alpha_{j,k,p,q,s}\phi^b_{j+k-s,p,l}\phi^f_{s,q,t}
\end{equation}
and
\begin{equation}
\label{betavgl}
\phi^b_{i,p,l}\phi^f_{s,q,t}=\sum_{k=0}^{\min(n-q,i+s)}\beta_{i,s,p,q,k}\phi_{i+s-k,k,p,q,l,t}
\end{equation}
for some real coefficients $\alpha_{j,k,p,q,s}$ and $\beta_{i,s,p,q,k}$ with $\alpha_{j,k,p,q,s}= (-1)^{k-s}\beta_{j+k-s,s,p,q,k}$.
\end{lemma}

\begin{proof}
Both the sets $\{\phi_{j,k,p,q,l,t}\}$ and $\{\phi^b_{\mu,\nu,\rho}\phi^f_{\lambda,\delta,\gamma}\}$ constitute a basis for $\cP\exp(-R^2/2)$. Therefore each $\phi_{j,k,p,q,l,t} $ can be expressed as a linear combination of elements of the set $\{\phi^b_{\mu,\nu,\rho}\phi^f_{\lambda,\delta,\gamma}\}$ and vice versa. The decomposition of the Grassmann algebra in lemma \ref{scalFischer} implies that $q=\delta$ and $t=\gamma$ are necessary conditions for the coefficients not to be zero. The coefficients $\alpha$ can be calculated using the inner product $\langle \cdot|\cdot\rangle_1$ for which the product basis is orthonormal. The bosonic integration implies that for this to be different from zero the coefficients must satisfy $p=\nu$ and $l=\rho$. Since the Hamiltonian for the harmonic oscillator is hermitian with respect to $\langle\cdot|\cdot\rangle_1$ another necessary condition is $2j+2k+p+q=2\mu+\nu+2\lambda+\delta$ which leads to $j+k=\mu+\lambda$. 

From these considerations the proposed summation is obtained. In order to find the relation between $\alpha_{j,k,p,q,s}$ and $\beta_{j,s,p,q,k}$ we calculate (we omit the coefficients $t$ and $l$ as they are not important)

\[\alpha_{j,k,p,q,s}=\langle \phi_{j,k,p,q,l,t}|\phi_{j+k-s,p,l}^b\phi^f_{s,q,t}\rangle_1=\]
\begin{equation}
\label{defalpha}
\int_{\mR^{m|2n}}\frac{L_j^{\frac{M}{2}+2k+p+q-1}(R^2)\,f_{k,p,q}H_p^{b}H_q^{f}\,L_{j+k-s}^{\frac{m}{2}+p-1}(r^2)H_p^{b}(-1)^s L_{s}^{q-n-1}(\theta^2)\widetilde{H}_q^{f}\exp(-R^2)}{\zeta^M_{j,2k+p+q}\,\sqrt{ a_{k,p,q}\, b_{k,p,q}}\quad\zeta^m_{j+k-s,p}\,\zeta^f_{s,q}}
\end{equation}
\[
=(-1)^s\int_{\mR^{m|2n}}\frac{L_{j+k-s}^{\frac{m}{2}+p-1}(r^2)H_p^{b}  L_{s}^{q-n-1}(\theta^2)H_q^{f}L_j^{\frac{M}{2}+2k+p+q-1}(R^2)\,f_{k,p,q}H_p^{b}\widetilde{H}_q^{f}\exp(-R^2)}{\zeta^m_{j+k-s,p}\,\zeta^f_{s,q}\quad\zeta^M_{j,2k+p+q}\, \sqrt{a_{k,p,q}\, b_{k,p,q}}}
\]
\[=(-1)^{s-k}\langle\phi^b_{j+k-s,p,l}\phi^f_{s,q,t}|\phi_{j,k,p,q,l,t}\rangle_2=(-1)^{s-k}\beta_{j+k-s,s,p,q,k}\]
which concludes the proof.
\end{proof}

It is our aim to show that the coefficients $\alpha$ and $\beta$ satisfy a polynomial bound. This will be useful to construct extensions of the space on which we can define $\langle\cdot|\cdot\rangle_2$ and to study the corresponding Hilbert space. First we need the following lemmas.
\begin{lemma}
For $H_j\in\cH_j$ and $P\in\cP$ of the form $P=\sum_{l=0}^jP_l$ with $\mE P_l=lP_l$, the following relation holds,
\label{helporth}
\begin{eqnarray*}
\int_{\mR^{m|2n}}H_j P\exp(-R^2)&=&\frac{1}{2}\Gamma(\frac{M}{2}+j)\int_{SS}H_jP\\
&=&\frac{1}{2}\Gamma(\frac{M}{2}+j)\int_{SS}H_jP_j.
\end{eqnarray*}
\end{lemma}
\begin{proof}
The Fischer decomposition in lemma \ref{scalFischer} implies $P_l=\sum_{j=0}^{\lfloor l/2\rfloor}R^{2j}H_{l-2j}^l$ with $H_{l-2j}^l\in\cH_{l-2j}$.
Applying equation \eqref{superint2} and the properties in theorem \ref{Pizzettithm} yields
\begin{eqnarray*}
\int_{\mR^{m|2n}}H_j P\exp(-R^2)&=&\sum_{l=0}^j\frac{1}{2}\Gamma(\frac{M+j+l}{2})\int_{SS}H_jP_l\\
&=&\frac{1}{2}\Gamma(\frac{M}{2}+j)\int_{SS}H_jH_{j}^j\\
&=&\frac{1}{2}\Gamma(\frac{M}{2}+j)\int_{SS}H_jP_j,
\end{eqnarray*}
which proves the lemma.
\end{proof}

\begin{lemma}
\label{techSchw}
For all $m,n,k,p,q,s\in\mN$ with $s>0$, $n-s-q\ge0$ and $M=m-2n>0$ and for $\lambda=\sqrt{1/s(n+1)}$ there exists an $N(n,q,s)\in\mN$ depending on $n,q$ and $s$ for which
\begin{eqnarray*}
\left(\sqrt{\frac{j+k-s}{j}\frac{j+p+k-s+\frac{m}{2}-1}{j+2k+p+q+\frac{M}{2}-1}}+\lambda\sqrt{\frac{s(n-s-q+1)}{j(j+2k+p+q+\frac{M}{2}-1)}}\right)\left(\frac{j-1}{j}\right)^{\frac{n-q-2s+2}{2}}&\le&1
\end{eqnarray*}
when $j>N$.
\end{lemma}
\begin{proof}
The left-hand side is smaller than
\begin{eqnarray*}
\sqrt{\left(1+\frac{k-s}{j}\right)\left(1+\frac{n-s-k-q}{j-1}\right)\left(1-\frac{1}{j}\right)^{n-q-2s+2}}+\lambda\sqrt{\frac{s(n-s-q+1)}{j(j-1)}}
\end{eqnarray*}
which can be expanded in $1/j$,
\begin{eqnarray*}
1-(1-\sqrt{\frac{n-s-q+1}{n+1}})\frac{1}{j}+\mathcal{O}(\frac{1}{j^2}).
\end{eqnarray*}
Since $(1-\sqrt{\frac{n-s-q+1}{n+1}})>0$, there exists an $N(n,q,s)\in\mN$, such that this expression is smaller than $1$ for $j>N$.
\end{proof}

\begin{theorem}
\label{afschatting1}
There exists a constant $C>0$, independent of $(j,k,p,q,s)$, such that for the $\alpha$ introduced in lemma \ref{basisovergangsbf}, the relation
\begin{eqnarray*}
|\alpha_{j,k,p,q,s}|\le C\sqrt{j^{n-q-2s+2}p^{s-k}}
\end{eqnarray*}
holds for all $j,k,p,q,s$.
\end{theorem}
\begin{proof}
First we calculate $\alpha_{0,k,p,q,s}$ starting from equation \eqref{defalpha} using lemma \ref{helporth} with $H_{2k+p+q}=f_{k,p,q}H_p^{b}H_q^{f}$ and $P=L_{k-s}^{\frac{m}{2}+p-1}(r^2)H_p^{b}(-1)^s L_s^{q-n-1}(\theta^2)\widetilde{H}_q^{f}$,
\begin{eqnarray*}
\alpha_{0,k,p,q,s}&=&\int_{\mR^{m|2n}}\frac{f_{k,p,q}H_p^{b}H_q^{f}\,L_{k-s}^{\frac{m}{2}+p-1}(r^2)H_p^{b}(-1)^s L_s^{q-n-1}(\theta^2)\widetilde{H}_q^{f}\exp(-R^2)}{\zeta^M_{0,2k+p+q}\,\sqrt{ a_{k,p,q}\, b_{k,p,q}}\quad\zeta^m_{k-s,p}\zeta^f_{s,q}}\\
&=&(-1)^{k-s}\frac{\Gamma(\frac{M}{2}+2k+p+q)}{2(k-s)!s!}\int_{SS}\frac{f_{k,p,q}H_p^{b}H_q^{f}\,r^{2k-2s}\theta^{2s}H_p^{b} \widetilde{H}_q^{f}}{\zeta^M_{0,2k+p+q}\,\sqrt{ a_{k,p,q}\, b_{k,p,q}}\zeta^m_{k-s,p}\zeta^f_{s,q}}.
\end{eqnarray*}
Using the first property in theorem \ref{Pizzettithm}, lemma \ref{helporth} and equation \eqref{bkpq} then yields
\begin{eqnarray*}
\alpha_{0,k,p,q,s}&=&(-1)^k\frac{\Gamma(\frac{M}{2}+2k+p+q)}{2(k-s)!s!}\int_{SS}\frac{f_{k,p,q}H_p^{b}H_q^{f}\,r^{2k-2s}(r^2-1)^{s}H_p^{b} \widetilde{H}_q^{f}}{\zeta^M_{0,2k+p+q}\,\sqrt{ a_{k,p,q}\, b_{k,p,q}}\zeta^m_{k-s,p}\zeta^f_{s,q}}\\
&=&(-1)^k\frac{\Gamma(\frac{M}{2}+2k+p+q)}{2(k-s)!s!}\int_{SS}\frac{r^{2k}f_{k,p,q}H_p^{b}H_q^{f}\,H_p^{b} \widetilde{H}_q^{f}}{\zeta^M_{0,2k+p+q}\,\sqrt{ a_{k,p,q}\, b_{k,p,q}}\zeta^m_{k-s,p}\zeta^f_{s,q}}\\
&=&(-1)^k\frac{\Gamma(\frac{M}{2}+2k+p+q)}{2(k-s)!s!}\frac{b_{k,p,q}}{\zeta^M_{0,2k+p+q}\,\sqrt{ a_{k,p,q}\, b_{k,p,q}}\zeta^m_{k-s,p}\zeta^f_{s,q}}\\
&=&(-1)^k\sqrt{\binom{k}{s}\frac{\Gamma(\frac{m}{2}+p+k)}{\Gamma(\frac{m}{2}+p+k-s)}\frac{\Gamma(\frac{M}{2}+p+q+k-1)}{\Gamma(\frac{M}{2}+p+q+2k-1)}\frac{(n-q-s)!}{(n-q-k)!}}.
\end{eqnarray*}
Because $k,q$ and $s$ are bounded (by $n$) there exists a constant $C^*$ such that
\begin{eqnarray*}
\alpha_{0,k,p,q,s}&\le&C^*\sqrt{p^{s-k}},
\end{eqnarray*}
for all $k,p,q,s$. Using lemma \ref{eigCH} and the properties \eqref{adj1} of $\langle \cdot|\cdot\rangle_1$ we calculate,
\begin{eqnarray*}
& &-4\sqrt{j(j+2k+p+q+\frac{M}{2}-1)}\alpha_{j,k,p,q,s}\\
&=&\langle(\nabla^2+R^2-2\mE-M)\phi_{j-1,k,p,q,l,t}|\phi^b_{j+k-s,p,l}\phi^f_{s,q,t}\rangle_1\\
&=&\langle\phi_{j-1,k,p,q,l,t}|(\nabla_b^2+r^2+2\mE_b+m)\phi^b_{j+k-s,p,l}\phi^f_{s,q,t}\rangle_1\\
&-&\langle\phi_{j-1,k,p,q,l,t}|(\theta^2+\nabla_f^2+2\mE_f-2n)\phi^b_{j+k-s,p,l}\phi^f_{s,q,t}\rangle_1.
\end{eqnarray*}
Applying lemma \ref{eigCH}, now for the purely bosonic case and equation \eqref{annferm}, yields
\begin{eqnarray}
\label{recursalpha}
\alpha_{j,k,p,q,s}&=&\sqrt{\frac{(j+k-s)(j+k-s+p+\frac{m}{2}-1)}{j(j+2k+p+q+\frac{M}{2}-1)}}\alpha_{j-1,k,p,q,s}\\
\nonumber
&+&\sqrt{\frac{s(n-s-q+1)}{j(j+2k+p+q+\frac{M}{2}-1)}}\alpha_{j-1,k,p,q,s-1}.
\end{eqnarray}
For $s=0$ this immediately yields
\begin{equation}
\label{alphaexpl}
\alpha_{j,k,p,q,0}=\sqrt{\binom{j+k}{j}\frac{\Gamma(j+k+p+\frac{m}{2})\Gamma(2k+p+q+\frac{M}{2})}{\Gamma(k+p+\frac{m}{2})\Gamma(j+2k+p+q+\frac{M}{2})}} \alpha_{0,k,p,q,0}.
\end{equation}
So there are constants $C_0$ and $C'_0$ such that
\begin{eqnarray*}
|\alpha_{j,k,p,q,0}|&\le&C'_0\sqrt{j^k(j+p)^{n-q-k}p^{k+q-n}}|\alpha_{0,k,p,q,0}|\\
&\le&C'_0C^\ast\sqrt{j^k(j+p)^{n-q-k}p^{k+q-n}}\sqrt{p^{-k}}\\
&\le&C_0\sqrt{j^{n-q}p^{-k}} \le C_0\sqrt{j^{n-q+2}p^{-k}},
\end{eqnarray*}
since $k\le n-q$, for all $j,k,p,q$. Now we take $N^\ast(n)$, the well-defined maximum of a finite set of $N(n,q,s)$ from lemma \ref{techSchw} 
\begin{eqnarray*}
N^\ast(n)&=&\max_{s\ge1,q|s+q\le n}N(n,q,s).
\end{eqnarray*}
For each $s$, $1\le s\le n$ we define $C_s$ as
\begin{eqnarray*}
C_s=\max\left(\sqrt{s(n+1)}\,C_{s-1},\sup_{j\le N^\ast ,k,p,q}\frac{|\alpha_{j,k,p,q,s}|}{\sqrt{j^{n-q-2s+2}p^{s-k}}}\right).
\end{eqnarray*}
The supremum is well-defined since $k$ and $q$ are bounded and because of the recursion relation \eqref{recursalpha}. In particular the relation
\begin{eqnarray*}
|\alpha_{j,k,p,q,s}|&\le &C_s\sqrt{j^{n-q-2s+2}p^{s-k}} \mbox{   for } j\le N^\ast(n)
\end{eqnarray*}
holds. Now we can prove that $|\alpha_{j,k,p,q,s}|\le C_s\sqrt{j^{n-q-2s+2}p^{s-k}}$ holds for all $j$. We do this by induction on $j$ for $j>N^\ast(n)$. If it holds for $j-1$, then
\begin{eqnarray*}
|\alpha_{j,k,p,q,s}|&\le&\sqrt{\frac{(j+k-s)(j+k-s+p+\frac{m}{2}-1)}{j(j+2k+p+q+\frac{M}{2}-1)}}C_s\sqrt{(j-1)^{n-q-2s+2}p^{s-k}}\\
&+&\sqrt{\frac{s(n-s-q+1)}{j(j+2k+p+q+\frac{M}{2}-1)}}C_{s-1}\sqrt{(j-1)^{n-q-2s+2}p^{s-k}}\\
&\le&C_s[\sqrt{\frac{(j+k-s)(j+k-s+p+\frac{m}{2}-1)}{j(j+2k+p+q+\frac{M}{2}-1)}}\\
&+&\sqrt\frac{1}{s(n+1)}\sqrt{\frac{s(n-s-q+1)}{j(j+2k+p+q+\frac{M}{2}-1)}}]\sqrt{(j-1)^{n-q-2s+2}p^{s-k}}\\
&\le&C_s\sqrt{j^{n-q-2s+2}p^{s-k}}
\end{eqnarray*}
by lemma \ref{techSchw}. So we obtain the theorem for $C=C_{n}$.
\end{proof}

Since $q\ge 0$ and $s\le n$ this result implies $|\alpha_{j,k,p,q,s}|\le C (j+p)^{\frac{n}{2}+1}$ for all $(j,k,p,q,s)$. This leads to the following corollary.

\begin{corollary}
There exists a constant $D>0$ such that for the $\alpha$ and $\beta$ introduced in lemma \ref{basisovergangsbf}
\begin{eqnarray*}
|\alpha_{j,k,p,q,s}|&\le& D\sqrt{(2j+2k+p+q+1)^{n+2}}
\end{eqnarray*}
and
\begin{eqnarray*}
|\beta_{i,s,p,q,k}|&\le& D\sqrt{(2i+2s+p+q+1)^{n+2}}
\end{eqnarray*}
holds for all $i,j,k,p,q,s$. The constant $D$ is independent of $i,j,k,p,q,s$.
\end{corollary}
From this corollary we immediately obtain the following lemma.

\begin{lemma}
\label{afschatting}
If a function $f\in L_2(\mR^m)_{m|2n}$ can be expanded as
\begin{eqnarray*}
\sum_{j,k,p,q,l,t}d_{j,k,p,q,l,t}\phi_{j,k,p,q,l,t}&and \quad as&\sum_{i,p,l,s,q,t}c_{i,p,l,s,q,t}\phi^b_{i,p,l}\phi^f_{s,q,t}
\end{eqnarray*}
in the $L_2(\mR^m)\otimes\Lambda_{2n}$-topology, then
\[
|c_{i,p,l,s,q,t}| \le nD\left(2i+2s+p+q+1\right)^{\frac{n}{2}+1}\max_{j,k|j+k=i+s}|d_{j,k,p,q,l,t}|
\]
and
\[
|d_{j,k,p,q,l,t}| \le nD\left(2j+2k+p+q+1\right)^{\frac{n}{2}+1}\max_{i,s|i+s=j+k}|c_{i,p,l,s,q,t}|.
\]
\end{lemma}

\section{Extension of $T$ and $\langle \cdot|\cdot\rangle_2$ }
\setcounter{equation}{0}

\label{extension}

The space $\mR[x_1,\ldots,x_m]\exp(-r^2/2)$ is dense in $L_2(\mR^m)$ and $\cS(\mR^{m})$. So we can try to extend the inner product $\langle \cdot|\cdot\rangle_2$ to the spaces $L_2(\mR^m)\otimes\Lambda_{2n}$ and $\cS(\mR^{m})\otimes\Lambda_{2n}$ via Hahn-Banach. For $L_2(\mR^m)_{m|2n}$ this does not work because it can be proven that $\langle\cdot|\cdot\rangle_2$ is not continuous with respect to the $L_2$ topology. For the Schwartz space, $\cS(\mR^m)_{m|2n}$, the same question will be answered positively. This is the main topic of this section.

\subsection{Extension to $L_2(\mR^m)_{m|2n}$}\hspace*{\fill} \\
We prove that $\langle \cdot|\cdot \rangle_2$ cannot be extended to $L_2(\mR^m)_{m|2n}$.

\begin{theorem}
\label{nogoL2}
The bilinear product $\langle\cdot|\cdot\rangle_2$ on $\cP\exp(-R^2/2)$ cannot be continuously extended to $L_2(\mR^m)_{m|2n}$.
\end{theorem}
\begin{proof}
Consider the function 
\begin{eqnarray*}
f(\bold{x})&=&\sum_{i=0}^\infty\frac{1}{i}\sqrt{\frac{2i!}{\Gamma(i+\frac{m}{2}+k)}}L_i^{\frac{m}{2}+k-1}(r^2)\sqrt\frac{\Gamma(m/2)}{2\pi^{m/2}}\exp(-r^2/2)\pi^{n/2}\exp(-\theta^2/2)\\
&=&\sum_{i=0}^\infty\frac{1}{i}\phi^b_{i,0,1}(\ux)\phi^f_{0,0,1}(\uxb)=\lim_{r\to\infty}\sum_{i=0}^r\frac{1}{i}\phi^b_{i,0,1}(\ux)\phi^f_{0,0,1}(\uxb)\\
&=&\lim_{r\to\infty}f_r(\bold{x}).
\end{eqnarray*}
The function $f(\bold{x})$ is clearly an element of $L_2(\mR^m)\otimes\Lambda_{2n}$ and $f_r(\bold{x})$ is an element of $\cP\exp(-R^2/2)$. Now we calculate
\begin{eqnarray*}
\langle f_r|f_r\rangle_2&=&\sum_{i,j=0}^r\frac{1}{ij}\langle \phi^b_{i,0,1}\phi^f_{0,0,1}|\phi^b_{j,0,1}\phi^f_{0,0,1}\rangle_2\\
&=&\sum_{i=0}^r\frac{1}{i^2}\langle \phi^b_{i,0,1}\phi^f_{0,0,1}|\phi^b_{i,0,1}\phi^f_{0,0,1}\rangle_2.
\end{eqnarray*}
The second equality follows from the hermiticity of the harmonic oscillator and equation \eqref{HOprod}. Formula \eqref{betavgl} and the fact that $\{\phi_{j,k,l}\}$ is $\langle\cdot|\cdot\rangle_2$-orthonormal yields
\begin{eqnarray*}
\langle \phi^b_{i,0,1}\phi^f_{0,0,1}|\phi^b_{i,0,1}\phi^f_{0,0,1}\rangle_2&=&\sum_{k=0}^{\min(n,i)}\beta_{i,0,0,0,k}\langle \phi_{i-k,k,0,0,1,1}|\phi^b_{i,0,1}\phi^f_{0,0,1}\rangle_2\\
&=&\sum_{k=0}^{\min(n,i)}|\beta_{i,0,0,0,k}|^2\\
&\ge&|\beta_{i,0,0,0,0}|^2=|\alpha_{i,0,0,0,0}|^2\\
&=&\frac{\Gamma(i+\frac{m}{2})\Gamma(\frac{M}{2})}{\Gamma(i+\frac{M}{2})\Gamma(\frac{m}{2})}\ge i^n\frac{\Gamma(\frac{M}{2})}{\Gamma(\frac{m}{2})}
\end{eqnarray*}
where we used equation \eqref{alphaexpl} in the last equality. This implies
\[
\langle f_r|f_r\rangle_2 \ge \frac{\Gamma(\frac{M}{2})}{\Gamma(\frac{m}{2})}\sum_{i=0}^ri^{n-2},
\]
so $\lim_{r\to\infty}\langle f_r|f_r\rangle_2$ diverges and $\langle \cdot|\cdot\rangle_2$ cannot be extended to $L_2(\mR^m)_{m|2n}$.
\end{proof}

\subsection{Spherical Hermite representation for the super Schwartz space}\hspace*{\fill} \\
The Schwartz space $\cS(\mR^m)$ is the space of infinitely differentiable, rapidly decreasing functions. This means that $||\ux^{\underline{\alpha}}\upx^{\underline{\beta}}f||_\infty$ is finite for all multi-indices $\underline{\alpha},\underline{\beta}$ when $f\in\cS(\mR^m)$. The topology on $\cS(\mR^m)$ is defined by the family of norms $||\ux^{\underline{\alpha}}\upx^{\underline{\beta}}.||_\infty$. It can be proven that this family of norms is equivalent with (generates the same topology as) the family of norms $||\ux^{\underline{\alpha}}\upx^{\underline{\beta}}\cdot||_{L_2(\mR^m)}$, see \cite{MR0751959}. Another family of norms which generates the topology of the Schwartz space $\cS(\mR^m)$ is given by 
\begin{equation}
\label{semnorm1}
||\cdot||_{\underline{\beta}}=||(\underline{N}+\underline{1})^{\underline{\beta}}\cdot||_{L_2(\mR^m)}\qquad \forall \underline{\beta}\in\mN^m,
\end{equation}
with $N_i=a_i^+ a_i^-$ (see formula (\ref{hamiltoniaan})) and $\underline{1}=(1,\cdots,1)$. This leads to the so-called Hermite-representation (or N-representation) theorem for the Schwartz space (see \cite{MR0751959, MR0278067}). Using this, it is not difficult to find another equivalent family of norms.

\begin{lemma}
\label{normenCH}
The family of norms $||\cdot||_{\underline{\beta}}$ on $\cS(\mR^m)$, as defined in equation (\ref{semnorm1}), is equivalent to the family of norms $||\cdot||_r$ ($r\in\mN$)
\[
||\cdot||_r=||(N^b+1)^r\cdot||_{L_2(\mR^m)}, \qquad N^b=\sum_{i=1}^mN_i.
\]
\end{lemma}
\begin{proof}
This follows from the inequalities $||\cdot||^2_{\underline{\beta}}\le ||\cdot||^2_{m|\underline{\beta}|}$ and $||\cdot||_k^2\le||\cdot||_{k\underline{1}}^2$. 
\end{proof}

For $f\in\cS(\mR^m)$ and $c_{j,k,l}=\langle \phi^b_{j,k,l}|f\rangle_{L_2(\mR^m)}$, the hermiticity of $N^b$ implies $\langle \phi^b_{j,k,l}|(N^b+1)^rf\rangle_{L_2(\mR^m)}=(2j+k+1)^rc_{j,k,l}$. Consequently, the norms in lemma \ref{normenCH} are given explicitly by
\begin{equation}
\label{CHreprbos}
||f||^2_{r}=\sum_{j,k,l}(2j+k+1)^{2r}|c_{j,k,l}|^2.
\end{equation}

Recall that the Schwartz space in superspace is defined as $\cS(\mR^{m})_{m|2n}=\cS(\mR^m)\otimes\Lambda_{2n}$, see \cite{DBS9}. The natural topology on this space is the product topology of $\cS(\mR^m)$ and $\Lambda_{2n}$.  From the definition of the topology, $f_k\to_{\cS(\mR^m)_{m|2n}} f$ is equivalent with $f_k^{(j)}\to_{\cS(\mR^m)}f^{(j)}$, with $\sum_{j=1}^{2^{2n}}f^{(j)}_ke_j=f_k$ for $e_j$ an arbitrary basis of the finite-dimensional vector space $\Lambda_{2n}$. We choose for $\{e_j\}$ the basis $\{\phi^f_{s,q,t}\}$, which is orthonormal with respect to the inner product $\langle \cdot|\cdot\rangle_{\Lambda_{2n}}$. We hence obtain the following lemma.

\begin{lemma}
\label{supertopologie}
The topology on $\cS(\mR^m)_{m|2n}$ is generated by the family of norms 
\[
||\cdot||_r(s,q,t)=||\langle \phi^f_{s,q,t}|.\rangle_{\Lambda_{2n}}||_r\qquad r,s,q,t\in\mN \mbox{ with } s+q\le n \mbox{ and } t\le\dim\cH_q^f,
\]
with $||\cdot ||_r$ the norms on $\cS(\mR^m)$ in lemma \ref{normenCH}.
\end{lemma}

Now we start to construct the spherical Hermite representation theorem on superspace.

\begin{lemma}
\label{eerstesupernorms}
The family of norms $||\cdot||^*_r$ ($r\in\mN$) on $\cS(\mR^m)_{m|2n}$,
\[
||\cdot||^*_r=\sqrt{\langle (N+1)^r\cdot|(N+1)^r.\rangle_1}
\]
with $N=N^b +\sum_{i=1}^{2n}b_i^+ b_i^-$, see equation (\ref{hamiltoniaan}), is equivalent with the family of norms in lemma \ref{supertopologie}. 
\end{lemma}

\begin{proof}
For an $f\in \cS(\mR^m)_{m|2n}$ we use equation (\ref{CHreprbos}) to calculate
\begin{eqnarray*}
||f||^{2}_r(s,q,t)&=&\sum_{i,p,l}(2i+p+1)^{2r}|\langle \phi^b_{i,p,l} \phi^f_{s,q,t}|f\rangle_1|^2\\
&\le&\sum_{s',q',t'}\sum_{i,p,l}(2i+p+2s'+q'+1)^{2r}|\langle \phi^b_{i,p,l} \phi^f_{s',q',t'}|f\rangle_1|^2\\
&=&\sum_{s',q',t',i,p,l}|\langle (N+1)^r \phi^b_{i,p,l} \phi^f_{s',q',t'}|f\rangle_1|^2\\
&=&\sum_{s',q',t',i,p,l}\langle (N+1)^rf |\phi^b_{i,p,l} \phi^f_{s',q',t'} \rangle_1\langle \phi^b_{i,p,l} \phi^f_{s',q',t'}| (N+1)^rf\rangle_1\\
&=&||f||_r^{*2}.
\end{eqnarray*}
In this calculation we used that $N = H - M/2$ and hence symmetric as $H$ is symmetric on the Hilbert space $L_2(\mR^m)\otimes\Lambda_{2n}$ (with inner product $\langle\cdot|\cdot\rangle_1$).

To complete the proof we also calculate
\begin{eqnarray*}
||f||_r^{*2}&\le&\sum_{s,q,t}\sum_{i,p,l}(2i+p+2n+1)^{2r}|\langle \phi^b_{i,p,l} \phi^f_{s,q,t}|f\rangle_1|^2\\
&=&\sum_{t=0}^{2r}\binom{2r}{t}(2n)^t\sum_{s,q,t}\sum_{i,p,l}(2i+p+1)^{2r-t}|\langle \phi^b_{i,p,l} \phi^f_{s,q,t}|f\rangle_1|^2\\
&\le&\sum_{t=0}^{2r}\binom{2r}{t}(2n)^t\sum_{s,q,t}\sum_{i,p,l}(2i+p+1)^{2r}|\langle \phi^b_{i,p,l} \phi^f_{s,q,t}|f\rangle_1|^2\\
&=&(1+2n)^{2r}\sum_{s,q,t}||f||^2_r(s,q,t),
\end{eqnarray*}
where the last sum $\sum_{s,q,t}$ is finite.
So we find that $\lim_{k\to\infty}||f||_r(s,q,t)=0$ for every $(r,s,q,t)$ is equivalent with $\lim_{k\to\infty}||f_k||_r^\ast=0$ for every $r$.
\end{proof}

If $f$ is an element of $\cS(\mR^m)_{m|2n}$ and $c_{i,p,l,s,q,t}=\langle \phi^b_{i,p,l}\phi^f_{s,q,t}|f\rangle_{1}$, then in the exact same way as equation (\ref{CHreprbos}) we obtain
\begin{eqnarray}
\label{superCHrepr}
||f||^{*2}_{r}=\sum_{i,p,l,s,q,t}(2i+2s+p+q+1)^{2r}|c_{i,p,l,s,q,t}|^2.
\end{eqnarray}

In the following we define a new set of norms on $\cS(\mR^m)_{m|2n}$, again with notation $||\cdot||_r$ like the norms on $\cS(\mR^m)$. This should not lead to confusion because they act on different spaces.

\begin{theorem}\label{superCHrepr2}
If $f$ is an element of $\cS(\mR^m)_{m|2n}$ and 
\begin{eqnarray*}
\sum_{j,k,p,q,l,t}d_{j,k,p,q,l,t}\phi_{j,k,p,q,l,t}=_{L_2(\mR^{m})_{m|2n}}f,
\end{eqnarray*}
then the family of norms $\{||\cdot||_r\}$, defined by
\begin{eqnarray*}
||f||^{2}_{r}=\sum_{j,k,p,q,l,t}(2j+2k+p+q+1)^{2r}|d_{j,k,p,q,l,t}|^2
\end{eqnarray*}
is equivalent with the family of norms in lemma \ref{eerstesupernorms} and thus generates the topology of $\cS(\mR^m)_{m|2n}$.
\end{theorem}

\begin{proof}
From lemma \ref{afschatting} we obtain with $c_{i,p,l,s,q,t}=\langle \phi^b_{i,p,l}\phi^f_{s,q,t}|f\rangle_{1}$
\begin{eqnarray*}
\sum_{i,s|i+s=\lambda}|c_{i,p,l,s,q,t}|^2&\le&D^2n^3(2\lambda+p+q+1)^{n+2}\max_{j,k|j+k=\lambda}|d_{j,k,p,q,l,t}|^2\\
&\le&D^2n^3(2\lambda+p+q+1)^{2\lfloor \frac{n+3}{2}\rfloor}\sum_{j,k|j+k=\lambda}|d_{j,k,p,q,l,t}|^2.
\end{eqnarray*}
This leads to $||f||^{2}_{r}\le n^3D^2||f||^{\ast 2}_{r+\lfloor \frac{n+3}{2}\rfloor}$ and similarly  $||f||^{\ast2}_{r}\le n^3D^2||f||^{ 2}_{r+\lfloor \frac{n+3}{2}\rfloor}$.
\end{proof}

Although this theorem, or the subsequent corollary \ref{CHrepr} appears to be very similar to the bosonic result in lemma \ref{normenCH} and equation \eqref{CHreprbos} it is highly non-trivial. Contrary to the bosonic case, the Hermite functions in superspace are not orthogonal with respect to the inner product $\langle\cdot|\cdot\rangle_1$ which defines the topology on $L_2(\mR^m)\otimes\Lambda_{2n}$. 
\begin{corollary} \label{CHrepr} (Spherical Hermite representation for $\cS(\mR^m)_{m|2n}$)
\\
Let $f$ be an element of $\cS(\mR^m)\otimes\Lambda_{2n}$ with $M>0$ which can be expanded as
\begin{eqnarray*}
\sum_{j,k,l}a_{j,k,l}\phi_{j,k,l}
\end{eqnarray*}
in the topology of $L_2(\mR^m)\otimes\Lambda_{2n}$ for $\{H_k^{(l)}\}$ satisfying equation \eqref{orthbasissuper}. The family of norms $\{||\cdot||_r|r\in\mN\}$ defined in theorem \ref{superCHrepr2} satisfy the relation
\begin{eqnarray*}
||f||^2_r& =&\sum_{j,k,l}|a_{j,k,l}|^2(2j+k+1)^{2r}.
\end{eqnarray*}
\end{corollary}

\begin{proof}
If $\{H_k^{(l)}\}$ is a basis of spherical harmonics is of the form $f_{i,p,q}H_p^{(l)}H_q^{(t)}$ then the corollary follows from theorem \ref{superCHrepr2}. Now we consider a general basis $\{H_k^{(l)}\}$ satisfying equation \eqref{orthbasissuper} and a fixed basis of the form $f_{i,p,q}H_p^{(l)}H_q^{(t)}$ satisfying equation \eqref{orthbasissuper} which we denote by $H_k^{'(l)}$. The corresponding Hermite functions are denoted respectively by $\phi_{j,k,l}$ and $\phi'_{j,k,l}$.

The basis transformation $\phi_{j,k,l}=\sum_{l}c_{lt}\phi'_{j,k,t}$ satisfies $\sum_{t}c_{lt}\overline{c}_{st}=\delta_{ls}$ since both the bases satisfy \eqref{orthbasissuper}. If $f=\sum_{j,k,l}a_{j,k,l}\phi_{j,k,l}$ then $f=\sum_{j,k,l,t}a_{j,k,l}c_{lt}\phi'_{j,k,t}$, so
\begin{eqnarray*}
||f||_r^2&=&\sum_{j,k,t}|\sum_l a_{j,k,l}c_{lt}|^2(2j+k+1)^{2r}\\
&=&\sum_{j,k,l} |a_{j,k,l}|^2(2j+k+1)^{2r},
\end{eqnarray*}
which proves the lemma.
\end{proof}

\subsection{Extension to $\cS(\mR^m)_{m|2n}$}\hspace*{\fill} \\
The inner product $\langle\cdot|\cdot\rangle_2$ can be extended to the super Schwartz space. 
\begin{definition}
\label{inS}
Let $f$ and $g$ be elements of $\cS(\mR^m)_{m|2n}$. Let $(f_i)_{i\in\mN}$ and $(g_s)_{s\in\mN}$ be sequences in $\cP\exp(-R^2/2)$ for which $\lim_{i\to\infty}f_i=f$ and $\lim_{s\to\infty}f_s=f$ in $\cS(\mR^m)_{m|2n}$, then $\langle f|g\rangle_2$ is defined by
\[
\langle f|g\rangle_2=\lim_{i,s\to\infty}\langle f_i|g_s\rangle_2.
\]
\end{definition}

The following lemma proves this is well-defined.
\begin{lemma}
The expression for $\langle f|g\rangle_2$ in definition \ref{inS} is finite and independent from the choice of the sequences $f_i$ and $g_s$.
\end{lemma}
\begin{proof}
Since $f_i$ and $g_s$ are elements of $\cP\exp(-R^2/2)$ they have (finite) expansions
\begin{eqnarray*}
f_i=\sum_{j,k,l}a^i_{j,k,l}\phi_{j,k,l}&\mbox{and}&g_s=\sum_{j,k,l}b^s_{j,k,l}\phi_{j,k,l}.
\end{eqnarray*}
Theorem \ref{orthCH} and definition \ref{inS} then yield
\[
\langle f|g\rangle_2=\lim_{i,s\to\infty}a^i_{j,k,l}\overline{b^s_{j,k,l}}.
\]
Since $\cS(\mR^m)_{m|2n}$-convergence implies convergence of the $||\cdot||_0$-norm in corollary \ref{CHrepr}, it is easily checked that this expression is finite and does not depend on the choice of the sequences.
 \end{proof}

This implies that the family of norms on $\cS(\mR^m)_{m|2n}$ in corollary \ref{CHrepr} can be expressed using the $\langle\cdot|\cdot\rangle_2$-inner product. Using the orthogonality of the super Hermite functions and the hermiticity of $N$ we obtain
\begin{equation}
\label{normin2}
||\cdot||_r^2=\langle (N+1)^r\cdot|(N+1)^r\cdot\rangle_2.
\end{equation}

In the following we prove that the inner product in definition \ref{inS} is still given by an expression similar to the one in theorem \ref{defsuper}.

\begin{theorem}
\label{Tctu}
The map $T:\cP\exp(-R^2/2)\to\cP\exp(-R^2/2)$ in formula \eqref{defT} is continuous with respect to the topology on $\cS(\mR^m)_{m|2n}$ and therefore has a unique extension to $\cS(\mR^m)_{m|2n}$.
\end{theorem}

\begin{proof}
Since we use a basis of $\cH_q^f$ such that $\widetilde{H_q^{f(t)}}=\pm i^qH_q^{f(t)}$ we find that the orthonormal basis of super Hermite functions $\{\phi_{j,k,l}\}$ consists of eigenvectors of $T$ with eigenvalues $\pm 1$ and $\pm i$. This implies that for $f\in\cP\exp(-R^2/2)$ expanded as the (finite) summation
\[
f=\sum_{j,k,l}d_{j,k,l}\phi_{j,k,l},
\]
the action of $T$ is given by
\[
T(f)=\sum_{j,k,l}d_{j,k,l}\lambda_{j,k,l}\phi_{j,k,l},
\]
with $|\lambda_{j,k,l}|=1$. Using corollary \ref{CHrepr} we find
\begin{eqnarray*}
||T(f)||^{2}_{r}&=&\sum_{j,k,l}(2j+k+1)^{2r}|\lambda_{j,k,l}d_{j,k,l}|^2\\
&=&||f||^{2}_{r}.
\end{eqnarray*}
This means that $\lim_{j\to\infty}f_j =0$ implies $\lim_{j\to\infty}T(f_j)= 0$ for $\cS$ convergence in $\cP\exp(-R^2/2)$.
\end{proof}

\begin{definition}
The linear operator $T$ on $\cS(\mR^m)_{m|2n}$ is defined on $f\in\cS(\mR^m)_{m|2n}$ as
\[
T(f)=\lim_{k\to\infty}T(f_k),
\]
with $f_k\to_{\cS(\mR^m)_{m|2n}}f$ and $f_k\in\cP\exp(-R^2/2)$.
\end{definition}

The fact that $T(f)$ does not depend on the choice of $\{f_k\}$ follows immediately from theorem \ref{Tctu}. Now we can write the inner product in definition \ref{inS} as an integral,
\begin{eqnarray*}
\langle f|g\rangle_2&=&\lim_{i,s\to\infty}\int_{\mR^{m|2n}}f_iT(g_s)=\int_{\mR^{m|2n}}f\lim_{s\to\infty}T(g_s)\\
&=&\int_{\mR^{m|2n}}fT(g).
\end{eqnarray*}

The limits can be brought into the integration as $f_i\to f$ and $g_s\to g$ in $\cS(\mR^m)_{m|2n}$ clearly implies that $f_ig_s\to fg$ in $L_1(\mR^{m})\otimes\Lambda_{2n}$.

\section{The $L_2$ Hilbert space in superspace}
\setcounter{equation}{0}

To study quantum mechanics on superspace (see \cite{DBS3, MR2395482, DBS8, MR967935}) properly, an inner product and a corresponding Hilbert space are necessary. As we would like orthosymplectically invariant Hamiltonians (like in \cite{DBS3, MR2395482, DBS8}) to be symmetric, $R^2$ and $\nabla^2$ should be symmetric. As a consequence of formula (\ref{commRD}) this implies that $\mE+M/2$ is skew-symmetric. So the realization of $\mathfrak{sl}_2$ by $iR^2/2$, $i\nabla^2/2$ and $\mE+M/2$ is skew-symmetric. It was shown in theorem 5.15 in \cite{CDBS2} that such an inner product can only exist when $M>0$, which is therefore the case we will consider in this section. In the general case (no restriction on $M$) there is the inner product $\langle\cdot|\cdot\rangle_1$, corresponding to the Hilbert space $L_2(\mR^m)\otimes \Lambda_{2n}=L_2(\mR^m)_{m|2n}$. 

The space of continuous linear functionals on the Schwartz space in superspace is easily seen to be 
\begin{eqnarray*}
(\cS(\mR^m)\otimes\Lambda_{2n})'= \cS'(\mR^m)\otimes \Lambda'_{2n}\cong \cS'(\mR^m)\otimes \Lambda_{2n}
\end{eqnarray*}
with $\cS'(\mR^m)$ the classical dual of the Schwartz space, the space of tempered distributions. So we obtain the rigged Hilbert space (\cite{MR0238534}) or Gelfand triple
\[\cS(\mR^{m})\otimes \Lambda_{2n}\subset L_2(\mR^m)\otimes \Lambda_{2n}\subset \cS'(\mR^m)\otimes \Lambda_{2n}.\] 
Unfortunately, $R^2$ is not symmetric with respect to $\langle\cdot|\cdot\rangle_1$ and there also are other undesirable properties, see the discussion in \cite{CDBS2}.

The inner product $\langle \cdot |\cdot \rangle_2$ in theorem \ref{defsuper} satisfies the desired properties. The Hilbert space corresponding to $\langle\cdot|\cdot\rangle_2$ will contain generalized functions. To prove that this is not a limitation of the $\langle\cdot|\cdot\rangle_2$-inner product we will show that this property must hold for every suitable Hilbert space. Therefore we will start by considering an arbitrary inner product $\langle\cdot|\cdot\rangle$ on which we will impose some straightforward conditions.

For an arbitrary Hilbert space in case $M>0$ for which $R^2$ and $\nabla^2$ are symmetric we consider the harmonic oscillator in superspace. The Hermite functions in formula (\ref{CH}) are its eigenvectors (lemma \ref{eigCH}) and should be a basis for the Hilbert space. The space of finite linear combinations of super Hermite functions corresponds to $\cP\exp(-R^2/2)$. Now we consider a general inner product $\langle\cdot|\cdot\rangle$ on $\cP\exp(-R^2/2)$ for which $R^2$ and $\nabla^2$ are symmetric operators. From the properties in lemma \ref{eigCH} we find that for any such inner product $\langle\cdot|\cdot\rangle$, the relation
\[
\langle\phi_{j,k,l}|\phi_{p,q,r}\rangle=\sqrt{\frac{p(p+q+\frac{M}{2}-1)}{j(j+k+\frac{M}{2}-1)}}\langle\phi_{j-1,k,l}|\phi_{p-1,q,r}\rangle
\]
holds. This implies $\langle \phi_{j,k,l}|\phi_{p,q,r}\rangle=0$ unless $j=p$ and $k=q$, so we find 
\begin{eqnarray*}
\langle \phi_{j,k,l}|\phi_{p,q,r}\rangle=\delta_{jp}\delta_{kq}\langle \phi_{0,k,l}|\phi_{0,k,r}\rangle.
\end{eqnarray*} 

Then, by definition, the corresponding Hilbert space $\cV_{\langle\cdot|\cdot\rangle}$ is the closure of $\cP\exp(-R^2/2)$ with respect to the topology induced on $\cP\exp(-R^2/2)$ by $\langle \cdot|\cdot\rangle$. We consider the `superfunction' determined by its formal series expansion 
\[
f(\bold{x})=\sum_{j=0}^\infty a_j\phi_{j,0,1}(\bold{x}),
\]
with $a_j\in\mC$. The Hermite functions $\phi_{j,0,1}$ are given by
\[
\phi_{j,0,1}(\bold{x})=\sqrt{\frac{j!\Gamma(\frac{M}{2})}{\pi^{M/2}\Gamma(j+\frac{M}{2})}}L_j^{\frac{M}{2}-1}(R^2)\,\exp(-R^2/2).
\]

Such a `function' $f(\bold{x})$ will belong to the Hilbert space $\cV_{\langle\cdot|\cdot\rangle}$ if and only if there exists a sequence of elements of $\cP\exp(-R^2/2)$ which converges to $f(\bold{x})$ in the $\langle\cdot|\cdot\rangle$-topology. If $f$ belongs to the Hilbert space, then $\langle f|f\rangle=\langle \phi_{0,0,1}|\phi_{0,0,1}\rangle\sum_{j=0}^\infty|a_j|^2<\infty$ must hold. This condition is necessary and sufficient. When $\sum_{j=0}^\infty|a_j|^2<\infty$ an example of such a sequence is given by the partial sums
\[
f_k(\bold{x})=\sum_{j=0}^k a_j\phi_{j,0,1}(\bold{x}).
\]

Now we consider the special case of the Euclidean space $\mR^M$ with $M$ bosonic variables, so for $\uy\in\mR^M$ and $\phi_{j,0,1}(\uy)$ as defined in equation \eqref{CH} or \eqref{CHbosM}. Since the Hermite functions are an orthonormal basis for the Hilbert space $L_2(\mR^M)$, we find that the condition of $\sum_{j=0}^\infty|a_j|^2<\infty$ is equal to demanding that 
\begin{eqnarray*}
f(\uy)=\sum_{j=0}^\infty a_j\phi_{j,0,1}(\uy)=\lim_{k\to\infty}f_k(\uy)
\end{eqnarray*}
is a function in $L_2(\mR^M)$. The $\phi_{j,0,1}(\uy)$ are radial functions, so we define $h_k$ by $f_k(\uy)=h_k(r_{\uy}^2)$ and 
\begin{eqnarray*}
\lim_{k\to\infty}h_k(r^2)=h(r^2) &\mbox{ in } L_2(\mR^+,r^{M-1}dr)
\end{eqnarray*}  
holds. It is clear that $f_k(\bold{x})=h_k(R^2)$ with $h_k(R^2)$ immediately defined as a polynomial in $R^2$ or as in definition \ref{sphsymm}.

Now in general $h(r^2)$ will not be a differentiable function which means that $h_k'(r^2)$ will not converge. So $f(\bold{x})$ will only be a formal summation or $h(R^2)$ is defined by definition \ref{sphsymm} in the weak sense (as we will see as an element of $\cS'(\mR^m)\otimes\Lambda_{2n}$). From these considerations we already find the following important fact.

\begin{proposition}
\label{Hilbertdist}
Every Hilbert space which contains the eigenvectors of the quantum harmonic oscillator on $\mR^{m|2n}$, with $n\not=0$, and for which $R^2$ and $\nabla^2$ are symmetric operators will contain formal summations which are not regular functions.
\end{proposition}

The inner product $\langle\cdot|\cdot\rangle_2$ in theorem \ref{defsuper} is an inner product with the desired properties and is the inner product we will use to construct the Hilbert space. Consider the bosonic harmonic oscillator. The function spaces $\mR[x_1,\cdots,x_m]\exp (-r^2/2)$ and $\cS(\mR^m)$ respectively contain states with finite energy and states with an infinitely small admixture of infinite energy states. The Hilbert space $L_2(\mR^m)$ also contains unphysical infinite energy states, see \cite{MR0238534}. In that sense the $L_2(\mR^m)$-space is nothing more than the mathematical completion of the Schwartz space with respect to the topology induced by the $L_2(\mR^m)$-inner product. Now, in superspace, this completion will not correspond to $L_2(\mR^m)_{m|2n}$, but to a new Hilbert space which we will denote as $\bold{L}_2(\mR^{m|2n})$. 

\begin{definition}
\label{definprod}
\label{Hilbert}
The Hilbert space $\bold{L}_2(\mR^{m|2n})$ is the space of formal series
\begin{eqnarray*}
\sum_{j,k,l}a_{j,k,l}\phi_{j,k,l}
\end{eqnarray*}
with $a_{j,k,l}\in\mC$ and $\sum_{j,k,l}|a_{j,k,l}|^2<\infty$ and $\phi_{j,k,l}$ defined in equation \eqref{CH} with the basis $\{H_k^{(l)}\}$ satisfying equation \eqref{orthbasissuper}. The inner product between the elements represented by the sequences $(a_{j,k,l})$ and $(b_{j,k,l})$ is given by
\[
\langle f|g\rangle_2 =\sum_{j,k,l}{a_{j,k,l}}\overline{b_{j,k,l}}.
\]
\end{definition}
In case $n=0$ this corresponds to the Hilbert space $L_2(\mR^m)$. 

\begin{remark}
The Hilbert space $\bold{L}_2(\mR^{m|2n})$ is a space of superfunctions with an inner product. This differs from the theory of so-called super Hilbert spaces (see e.g. \cite{MR1796030}) where one considers a Grassmann-valued inner product.
\end{remark}

Corollary \ref{CHrepr} implies that the super Schwartz space $\cS(\mR^{m})_{m|2n}$ is included in the Hilbert space $\bold{L}_2(\mR^{m|2n})$. The proof of theorem \ref{nogoL2} shows that $L_2(\mR^m)_{m|2n}$ is not included in this Hilbert space. Now we show that all elements of the Hilbert space, although not necessary regular functions, are tempered distributions.

\begin{lemma}
\label{HilbertSchwartz}
Generalized functions $f$ of the form
\begin{eqnarray*}
\sum_{j,k,l}a_{j,k,l}\phi_{j,k,l}\,\,&\mbox{with}&\,\,\sum_{j,k,l}|a_{j,k,l}|^2<\infty
\end{eqnarray*}
are elements of the space $\cS'(\mR^m)\otimes\Lambda_{2n}$ and therefore are derivatives of (almost everywhere) continuous functions with polynomial growth.
\end{lemma}
\begin{proof}
It is well-known that elements of $\cS'(\mR^m)$ are derivatives of (almost everywhere) continuous functions with polynomial growth, (see e.g. \cite{MR0209834, MR0278067}). The sequence $a_{j,k,l}$ is clearly bounded, so there exists a constant $C$ such that
\[
|a_{j,k,l}| \le C.
\]
We can express $f$ in the product basis $\phi^b_{i,p,l}\phi^f_{s,q,t}$ with purely bosonic and fermionic Hermite functions,
\[
f=\sum_{i,p,l,s,q,t}c_{i,p,l,s,q,t}\phi^b_{i,p,l}\phi^f_{s,q,t}.
\]
Lemma \ref{afschatting} implies there exists a constant $C^{\ast}$ for which
\[
|c_{i,p,l,s,q,t}|^2 \le C^{\ast}(2i+p+1)^{n+2}
\]
for all $i,p,l,s,q,t$ with $C^\ast=n^2D^2\sum_{j=0}^{n+2}\binom{n+2}{j}(2n)^jC^2$. This lemma can be used although $f$ is not in $L_2(\mR^m)_{m|2n}$ by using the partial sums of $f$, which are in $\cP\exp(-R^2/2)$. Now take a general $g(\ux)=\sum_{i,p,l}d_{i,p,l}\phi^b_{i,p,l}\in\cS(\mR^m)$. The formal expression
\begin{eqnarray*}
\left|\int_{\mR^m}\left(\sum_{i,p,l}c_{i,p,l,s,q,t}\phi^b_{i,p,l}(\ux)\right)g(\ux)dV(\ux)\right|^2&=&|\sum_{i,p,l}c_{i,p,l,s,q,t}d_{i,p,l}|^2\\
&\le& \sum_{i,p,l}C^{\ast}(2i+p+1)^{n+2}|d_{i,p,l}|^2\\
&=&C^{\ast}||g(\ux)||^2_{n+2}
\end{eqnarray*}
shows that the series $\sum_{i,p,l}c_{i,p,l,s,q,t}\phi^b_{i,p,l}(\ux)$ converges to a continuous linear functional on $\cS(\mR^m)$, which implies $f\in \cS'(\mR^m)\otimes\Lambda_{2n}$.
\end{proof}

The previous lemma implies that we obtain the Gelfand triple (\cite{MR0238534})
\begin{eqnarray*}
\cS(\mR^{m})\otimes \Lambda_{2n}\subset \bold{L}_2(\mR^{m|2n}) \subset \cS'(\mR^m)\otimes \Lambda_{2n}.
\end{eqnarray*}

By definition of the weak topology on $\cS'(\mR^m)$ we find the following lemma.
\begin{lemma}
\label{intecht}
For $f\in  \bold{L}_2(\mR^{m|2n})$ defined by the sequence $(a_{j,k,l})$, the following relation holds
\[
a_{j,k,l}=\int_{\mR^{m|2n}}f{T(\phi_{j,k,l})}
\]
where $\int_{\mR^{m|2n}}f{T(\phi_{j,k,l})}$ denotes the action of the tempered distribution $f$ with values in $\Lambda_{2n}$ on $T(\phi_{j,k,l})\in\cS(\mR^m)\otimes\Lambda_{2n}$.
\end{lemma}

The action of the derivatives and multiplication with variables on elements of the Hilbert space is already defined, as it is defined on $\cS'(\mR^m)\otimes\Lambda_{2n}$. 
\begin{definition}
\label{actieSp}
The action of elements of the algebra $Alg(X_i,\partial_{X_j})$ generated by the variables and derivatives on elements of $\cS'(\mR^m)\otimes \Lambda_{2n}$ is defined by the following rules:
\begin{itemize}
\item left multiplication with ${x\grave{}}_i$ and the derivation $\partial_{{x\grave{}}_j}$ on $\cS'(\mR^m)\otimes \Lambda_{2n}$ commute with $\cS'(\mR^m)$ and are defined on $\Lambda_{2n}$ in the standard way,
\item left multiplication with ${x}_i$ and the derivation $\partial_{{x}_j}$ on $\cS'(\mR^m)\otimes \Lambda_{2n}$ commute with $\Lambda_{2n}$ and are defined on $\cS'(\mR^m)$ in the standard way: for $f\in\cS'(\mR^m)$ and $g\in \cS(\mR^m)$,
\[
\int_{\mR^m}(x_if)gdV(\ux)=\int_{\mR^m}f(x_ig)dV(\ux)\mbox{ and }\int_{\mR^m}(\partial_{x_i}f)gdV(\ux)=-\int_{\mR^m}f(\partial_{x_i}g)dV(\ux).
\]
\end{itemize}
\end{definition}

\begin{definition}
Consider an $\cO\in Alg(X_i,\partial_{X_j})$. The elements $f\in\bold{L}_2(\mR^{m|2n})$ for which $\cO f$, as defined in definition \ref{actieSp} satisfy $\cO f\in\bold{L}_2(\mR^{m|2n})$, are said to be in the domain $\mD(\cO)$ of $\cO$. 
\end{definition}
The action of $\nabla^2$ can also be expressed using the fact that 
\[
\nabla^2\sum_{j,k,l}a_{j,k,l}\phi_{j,k,l}=\sum_{j,k,l}a_{j,k,l}\nabla^2\phi_{j,k,l}
\]
in $\cS'(\mR^{m})\otimes\Lambda_{2n}$ and the properties in lemma \ref{eigCH}, which lead to
\begin{eqnarray}
\nonumber
-\nabla^2\phi_{j,k,l}&=&(2j+k+\frac{M}{2})\phi_{j,k,l}\\
\label{expnabla2}
&+&\sqrt{(j+1)(j+\frac{M}{2}+k)}\phi_{j+1,k,l}+\sqrt{j(j+\frac{M}{2}+k-1)}\phi_{j-1,k,l}.
\end{eqnarray}

In particular, for regular functions for which the actions of the derivatives or of multiplication with variables exist, this definition coincides with the usual action. The elements of $\cP\exp(-R^2/2)$ are such functions, which implies that $\mD(\mathcal{O})$ is dense in $L_2(\mR^{m|2n})$ for $\mathcal{O}$ equal to $\nabla^2$, $\mE+M/2$ or multiplication with $R^2$.

\begin{theorem}
\label{hermitisch}
For the densely defined operators on $\bold{L}_2(\mR^{m|2n})$, $\nabla^2$, $\mE+M/2$ and multiplication with $R^2$, it holds that $\nabla^2$ and $R^2$ are symmetric, while $\mE+M/2$ is skew-symmetric.
\end{theorem}
\begin{proof}
We need to prove that for $f,g\in\mD(\nabla^2)$, $\langle \nabla^2 f|g\rangle_2=\langle f|\nabla^2 g\rangle_2$. This is easily seen to be true from the expression \eqref{expnabla2}. The proof for the other operators is similar.
\end{proof}

For use in section \ref{Heissection}, multiplication with
\begin{eqnarray*}
R&=&\sum_{j=0}^n\frac{\theta^{2j}}{j!}\frac{\Gamma(\frac{3}{2})}{\Gamma(\frac{3}{2}-j)}r^{1-2j}
\end{eqnarray*}
is needed. This is not defined everywhere on $\cS'(\mR^m)\otimes\Lambda_{2n}$. We start with the following definition.
\begin{definition}
\label{deffkl}
For an element $f\in\bold{L}_2(\mR^{m|2n})$ associated with the sequence $(a_{j,k,l})$ the following functions $f_{k,l}$ are associated:
\begin{eqnarray*}
f_{k,l}(u^2)=\sum_{j=0}^\infty a_{j,k,l}\frac{L_j^{\frac{M}{2}+k-1}(u^2)}{\zeta_{j,k}^M}\exp(-u^2/2) & & \in L_2(\mR^+,u^{M+2k-1}du).
\end{eqnarray*}
In the $\bold{L}_2(\mR^{m|2n})$-topology the following expression converges,
\begin{eqnarray*}
f(\bold{x})&=&\sum_{k,l}f_{k,l}(R^2)H_{k}^{(l)}.
\end{eqnarray*}
The element $f_{k,l}(R^2)H_{k}^{(l)}\in \bold{L}_2(\mR^{m|2n})$ is given by $\sum_{j}a_{j,k,l}\phi_{j,k,l}$, or represented by the sequence $c_{i,p,t}=\delta_{pk}\delta_{tl}a_{i,p,t}$.
\end{definition}

This leads to the definition of multiplication with $R$.
\begin{definition}
Let  $f\in \bold{L}_2(\mR^{m|2n})$ be given by $\sum_{k,l}f_{k,l}(R^2)H_k^{(l)}$. If $uf_{k,l}(u^2)\in L_2(\mR^+,u^{M-1}du)$ for each $k,l$ then 
\begin{eqnarray*}
Rf_{k,l}(R^2)H_{k}^{(l)} &\in &  \bold{L}_2(\mR^{m|2n})
\end{eqnarray*}
as in definition \ref{deffkl}. If the series converges in $\bold{L}_2(\mR^{m|2n})$ then $Rf$ is defined as
\begin{eqnarray*}
Rf(\bold{x})&=&\sum_{k,l}Rf_{k,l}(R^2)H_{k}^{(l)}.
\end{eqnarray*}
\end{definition}

This densely defined operator is symmetric.
\begin{lemma}
\label{Rsymm}
For $f,g\in \bold{L}_2(\mR^{m|2n})$ elements of $\mD(R)$, the relation
\begin{eqnarray*}
\langle R\,f|g\rangle_2&=&\langle f|R\,g\rangle_2
\end{eqnarray*}
holds.
\end{lemma}
\begin{proof}
The definition of the inner product and the notation in definition \ref{deffkl} leads to
\begin{eqnarray*}
\langle R\,f|g\rangle_2&=&\sum_{k,l}\langle (R\,f)_{k,l}(R^2)H_{k}^{(l)}|g_{k,l}(R^2)H_k^{(l)}\rangle_2\\
&=&\sum_{k,l}\langle R\, f_{k,l}(R^2)H_{k}^{(l)}|g_{k,l}(R^2)H_k^{(l)}\rangle_2.
\end{eqnarray*}
The lemma then follows from the observation
\begin{eqnarray*}
\langle h_1(R^2)H_{k}^{(l)}|h_2(R^2)H_k^{(l)}\rangle_2&=&\int_{\mR^+}h_1(u^2)h_2(u^2)u^{M+2k-1}du
\end{eqnarray*}
for $h_1(u^2),h_2(u^2)\in L_2(\mR^+,u^{M+2k-1}du)$.
\end{proof}

\begin{theorem}
\label{isom}
If a set of functions $f_{k,j}(r_{\uy}^2)H^{b(t)}_{M,k}(\uy)$, $j,k\in\mN$, $t=1,\cdots,\dim\cH_{M,k}^b$ constitutes a basis for $L_2(\mR^M)$, then the set of generalized functions 
\begin{eqnarray*}
\{f_{k,j}(R^2)H_k^{(l)}|j,k\in\mN,l=1,\cdots,\dim \cH_k\},
\end{eqnarray*}
with $\{H_k^{(l)}\}$ a basis for $\cH_k$, constitutes a basis for $\bold{L}_2(\mR^{m|2n})$.
\end{theorem}
\begin{proof}
When $f(r_{\uy}^2)H_{k}^{b(t)}(\uy)$ is an element of $L_2(\mR^M)$, then $f(R^2)H_k^{(l)}(\bold{x})$ is an element of $\bold{L}_2(\mR^{m|2n})$, this follows from the discussion before proposition \ref{Hilbertdist}. This also generates a morphism $\chi:L_2(\mR^M)_{k,t}\to \bold{L}_2(\mR^{m|2n})_{k,l}$, with $ \bold{L}_2(\mR^{m|2n})_{k,l}$ the subspace of $ \bold{L}_2(\mR^{m|2n})$ generated by $\{\phi_{j,k,l}(\bold{x}),j\in\mN\}$ with $k$ and $l$ fixed,
\[
\chi[f(r_{\uy}^2)H_{k}^{b(t)}(\uy)]=f(R^2)H_k^{(l)}(\bold{x}).
\]
This morphism satisfies $\chi[\phi^b_{j,k,t}(\uy)]=\phi_{j,k,l}(\bold{x})$ and therefore is an isomorphism. Since $\{f_{k,j}(r_{\uy}^2)H^{b(t)}_k(\uy),j\in\mN\}$ is a basis for $L_2(\mR^M)_{k,l}$, $\{f_{k,j}(R^2)H_k^{(l)}(\bold{x}),j\in\mN\}$ is a basis for $\bold{L}_2(\mR^{m|2n})_{k,l}$. This implies that every $\phi_{j,k,l}(\bold{x})$ can be expanded in terms of the $f_{k,j}(R^2)H_k^{(l)}(\bold{x})$.
\end{proof}

\begin{remark}
In view of the Gelfand triple, there are some elements of the Hilbert space $\bold{L}_2(\mR^{m|2n})$ which are usually regarded as generalized functions. There are also regular functions (elements of $L_2(\mR^m)\otimes\Lambda_{2n}$) which are only regarded as elements of $\cS'(\mR^m)\otimes\Lambda_{2n}$, but not as elements of the Hilbert space.
\end{remark}

This remark and the two Gelfand triples can be captured in the following venn diagram.

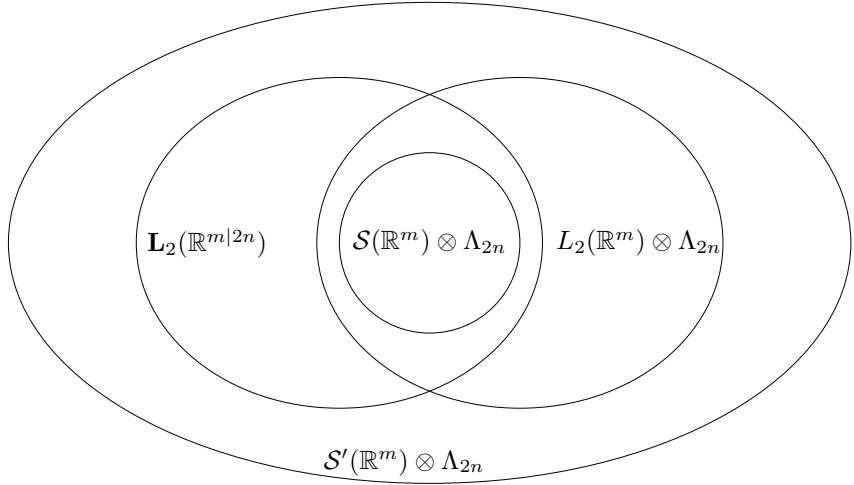
\def\firstcircle{(0:2.3cm) ellipse (2.7cm and 2.2cm)}
\def\secondcircle{(0:3.5cm) circle (1.2cm)}
\def\thirdcircle{(0:4.7cm) ellipse (2.7cm and 2.2cm)}
\def\fourthcircle{(0:3.5cm) ellipse (5.6cm and 3.2cm)}
\def\f{(-1.5cm:4.3cm) circle (0cm)}

\begin{tikzpicture}[line width=0.25pt]
\tikz[label distance = 2mm];

    \draw \firstcircle node[text=black] {$\bold{L}_2(\mR^{m|2n})\qquad\qquad\qquad\qquad\qquad$};
    \draw \secondcircle node [text=black] {$\cS(\mR^m)\otimes\Lambda_{2n}$};
    \draw \thirdcircle node [text=white,black] {$\qquad\qquad\qquad\qquad\quad L_2(\mR^m)\otimes\Lambda_{2n}$};
    \draw \fourthcircle;
    \draw \f node [text=black] {$\cS'(\mR^m)\otimes\Lambda_{2n}$};

\end{tikzpicture}

Now we consider the function $f(\bold{x})\in L_2(\mR^m)\otimes\Lambda_{2n}$ from the proof of theorem \ref{nogoL2}. Since $\langle f| f\rangle_2$ would be infinite we obtain $f(\bold{x})\not\in \bold{L}_2(\mR^{m|2n})$. In particular the venn diagram above shows this also implies $f(\bold{x})\not\in \cS(\mR^m)\otimes\Lambda_{2n}$. This can also be seen from corollary \ref{CHrepr} and equation \eqref{normin2}, $||f||_0$ is infinite. Since $f(\bold{x})\in L_2(\mR^m)\otimes\Lambda_{2n}$, the expression $||f||_0^\ast$ is finite, so the equivalence of norms implies that for some $r>0$, $||f||_r^\ast$ must be infinite.

As a final result of this section, we obtain a Parseval theorem for the super Fourier transform with respect to both inner products $\langle . |. \rangle_{1}$ and $\langle . |. \rangle_{2}$.

\begin{theorem}{(Parseval)}
\label{Pars}
The adjoint of $\cF_{m | 2n}^{\pm} $ under both of the inner products $\langle.|.\rangle_1$ and $\langle .|.\rangle_2$ is given by $\cF_{m | 2n}^{\mp} $. This implies that for $f$ and $g\in L_2(\mR^m)_{m|2n}$ one has
\[
\langle f|g\rangle_1=\langle \cF_{m | 2n}^{\pm} (f)|\cF_{m | 2n}^{\pm} (g)\rangle_1
\]
and for $f$ and $g\in \bold{L}_2(\mR^{m|2n})$
\[
\langle f|g\rangle_2=\langle \cF_{m | 2n}^{\pm} (f)|\cF_{m | 2n}^{\pm} (g)\rangle_2.
\]
\end{theorem}

\begin{proof}
In \cite{DBS9} it was proven that for $f, g \in L_2(\mR^m)_{m|2n}$, the following holds,
\begin{equation*}
\int_{\mR^{m|2n},\bold{x}} f(\bold{x}) \overline{g(\bold{x})} = \int_{\mR^{m|2n},\bold{y}} \cF_{m | 2n}^{\pm} (f)(\bold{y}) \overline{\cF_{m | 2n}^{\pm} (g)(\bold{y})}.
\end{equation*}
This implies that
\begin{eqnarray*}
\langle f|g\rangle_1&=&\int_{\mR^{m|2n},\bold{x}} f(\bold{x}) \ast\overline{g(\bold{x})}\\
&=& \int_{\mR^{m|2n},\bold{y}} \cF_{m | 2n}^{\pm} (f)(\bold{y}) \overline{\cF_{m | 2n}^{\pm} (\ast g)(\bold{y})}.
\end{eqnarray*}
In order to prove the first part of the theorem we therefore need to show that $\ast\cF_{m | 2n}^{\pm}(f)=\cF_{m | 2n}^{\pm}(\ast f)$ for $f$ in $L_2(\mR^m)_{m|2n}$.  Since $\cF_{m | 2n}^{\pm} = \cF_{m | 0}^{\pm} \circ \cF_{0 | 2n}^{\pm}$, this is equivalent to proving $\ast\cF_{0 | 2n}^{\pm}(\alpha)=\cF_{0 | 2n}^{\pm}(\ast \alpha)$ for $\alpha$ in $\Lambda_{2n}$. Equations \eqref{CHstar} and \eqref{FourCH} imply that $\ast$ and $\cF^\pm_{0|2n}$ have a coinciding basis of eigenvectors which proves they commute.

The Fourier transform on $\bold{L}(\mR^{m|2n})$ is the Hahn-Banach extension of the Fourier transform on $\cS(\mR^m)_{m|2n}$. The second part of the theorem is then immediately proven by the orthonormality of the super Hermite functions and equation \eqref{FourCH}.
\end{proof}

\section{Orthosymplectically invariant quantum problems}
\setcounter{equation}{0}

In this section we study orthosymplectically invariant Schr\"odinger equations in superspace. We prove that the solutions of the Schr\"odinger equations derived in \cite{CDBS3} form a complete set. We also derive a criterion for essential self-adjointness for orthosymplectically invariant Hamiltonians.

When spherically symmetric quantum Hamiltonians (such as the (an)harmonic oscillator \cite{DBS3} or the hydrogen atom \cite{MR2395482}) are generalized to superspace we get super Hamiltonians with an $\mathfrak{osp}(m|2n)$ invariance. In \cite{MR2395482} the energy eigenvalues and corresponding eigenspaces were determined for the quantum Kepler problem (hydrogen atom) in superspace. In \cite{DBS3} the basis of Hermite functions was constructed for the quantum harmonic oscillator in superspace. In \cite{CDBS3} general orthosymplectically invariant Schr\"odinger equations were studied in the context of orthosymplectically invariant functions and harmonic analysis. They were solved using the results from the purely bosonic case.

It was proven in \cite{CDBS3} that a general orthosymplectically invariant Hamiltonian is of the form
\begin{equation}
\label{superham}
H=-\frac{1}{2}\nabla^2+V(R^2),
\end{equation}
with $V(R^2)$ defined in definition \ref{sphsymm}. For every such a super Hamiltonian we can also consider the special case of the bosonic Hamiltonian of the form 
\begin{equation}
\label{bosham}
H_b=-\frac{1}{2}\nabla_{b,M}^2+V(r_{\uy}^2)
\end{equation}
in $M$ bosonic dimensions, $\uy\in\mR^M$ and $\nabla_{b,M}^2=\sum_{j=1}^M\partial_{y_j}^2$. With these notations, the following was proven in theorem 7 in \cite{CDBS3}.
 
\begin{lemma}
\label{superopl}
If the function $H_k^bf(r^2_{\uy})$ with $H_k^b\in\cH_{k,M}^b$ is an eigenvector of the $M$-dimensional Hamiltonian (\ref{bosham}) with eigenvalue $E$, then
\[
\left[-\frac{1}{2}\nabla^2+V(R^2)\right]H_k f(R^2)=EH_kf(R^2),
\]
for an arbitrary $H_k\in\cH_k$ in superspace $\mR^{m|2n}$. In case $f(r^2_{\uy})$ is $n$ times differentiable $f(R^2)$ is given by definition \ref{sphsymm}. If not, $f(R^2)$ is defined formally as element of $\bold{L}_2(\mR^{m|2n})$ by a series expansion (definition \ref{Hilbert}) or by a Taylor expansion as an element of $\cS'(\mR^{m})\otimes\Lambda_{2n}$ (definition \ref{sphsymm}).
\end{lemma}

Now the Hilbert space structure is obtained we can prove that this set of solutions is complete.
\begin{theorem}
\label{oplSchro}
If $f_{k,j}(r_{\uy}^2)H^{b(l)}_k$ is an orthonormal basis for $L_2(\mR^{M})$ of eigenvectors of the bosonic Schr\"odinger equation in $M$ dimensions (\ref{bosham}), then $f_{k,j}(R^2)H^{(l)}_k$ is an orthonormal basis for $\bold{L}_2(\mR^{m|2n})$ of eigenvectors of the Schr\"odinger equation
\[
\left[-\frac{1}{2}\nabla^2+V(R^2)\right]\chi(\bold{x})=E\chi(\bold{x}).
\]
\end{theorem}
\begin{proof}
This follows from lemma \ref{superopl} and theorem \ref{isom}.
\end{proof}

To find the multiplicities for the energy levels it is important to note that the dimension of the spherical harmonics of degree $k$ (in superspace) is given by
\[
\dim\cH_k=\sum_{i=0}^{\min(k,2n)}\binom{2n}{i}\binom{k-i+m-1}{m-1}-\sum_{i=0}^{\min(k-2,2n)}\binom{2n}{i}\binom{k-i+m-3}{m-1},
\]
see \cite{DBS5}.

The intersection of $L_2(\mR^m)_{m|2n}$ and $\bold{L}_2(\mR^{m|2n})$ has an important property for orthosymplectically invariant Schr\"odinger equations.
\begin{theorem}
\label{inter}
If the function $f(R^2)H_k$, with $f(R^2)$ as definition \ref{sphsymm} (so with $f\in C^n(\mR^+)$), is an element of $L_2(\mR^m)_{m|2n}$ for each $H_k\in\cH_k$, then
\begin{eqnarray*}
f(R^2)H_k&\in&\bold{L}_2(\mR^{m|2n}),
\end{eqnarray*}
for each $H_k\in\cH_k$.
\end{theorem}
\begin{proof}
Since $H_k$ is arbitrary, we consider $H_k=H_k^b\in\cH_k^b$, a normalized bosonic spherical harmonic. The condition $f(R^2)H_k^b\in L_2(\mR^m)_{m|2n}$ implies
\begin{eqnarray*}
f(R^2)f(R^2)H_k^bH_k^b&\in&L_1(\mR^m)_{m|2n}.
\end{eqnarray*}
This implies the following expression is finite
\begin{eqnarray*}
\int_{\mR^{m|2n}}f(R^2)f(R^2)H_k^bH_k^b&=&\int_{0}^\infty r^{m+2k-1}dr\int_Bf(R^2)f(R^2)
\end{eqnarray*}
The proof of theorem 4 in \cite{CDBS3} then shows this is equal to
\begin{eqnarray*}
c\int_{0}^\infty r^{M+2k-1}f^2(r^2)dr 
\end{eqnarray*}
for some coefficient $c$. This coefficient can be calculated by taking the example $f(R^2)=\exp(-R^2/2)$ and using equation \eqref{superint2} and lemma \ref{SSin4}
\begin{eqnarray*}
\int_{\mR^{m|2n}}\exp(-R^2)H_k^bH_k^b&=&\frac{1}{2}\Gamma(2k+\frac{M}{2})a_{0,k,0}b_{0,k,0}\\
&=&\frac{\Gamma(\frac{m}{2}+k)}{\Gamma(\frac{M}{2}+k)n!}\int_0^\infty  r^{M+2k-1}\exp(-r^2)dr,
\end{eqnarray*}
so $c=\frac{\Gamma(\frac{m}{2}+k)}{\Gamma(\frac{M}{2}+k)n!}$. Therefore 
\begin{eqnarray*}
\int_{0}^\infty r^{M+2k-1} f(r^2)&=&\frac{\Gamma(\frac{M}{2}+k)n!}{\Gamma(\frac{m}{2}+k)}\int_{\mR^{m|2n}}f(R^2)f(R^2)H_k^bH_k^b<\infty
\end{eqnarray*}
and $f(v^2)\in L_2(\mR^+,v^{M+2k-1}dv)$, from which it is clear that $f(R^2)H_k\in\bold{L}_2(\mR^{m|2n})$ for each $H_k\in \cH_k$ (see theorem \ref{isom}).
\end{proof}

\begin{remark}
This theorem is interesting for orthosymplectically invariant Schr\"odinger equations. If an eigenfunction of the typical form $f(R^2)H_k$ is found, which is an element of $L_2(\mR^m)_{m|2n}$, theorem \ref{inter} implies this is an actual solution, inside the Hilbert space. This is for instance the case for the solutions of the quantum Kepler problem obtained in \cite{MR2395482}.
\end{remark}
It is however still possible, in general, that there are solutions in $\bold{L}_2(\mR^{m|2n})$ which are not contained in $L_2(\mR^{m})_{m|2n}$.

We conclude this section with a criterion for essential self-adjointness for orthosymplectically invariant Hamiltonians in superspace. The proof is based on the classical case, theorem $X.11$ in \cite{MR0493420}.

\begin{theorem}
For $V\in C^n(\mR^+)$, the Hamiltonian 
\[
H=-\frac{1}{2}\nabla^2+V(R^2).
\]
on $\bold{L}_2(\mR^{m|2n})$ is essentially self-adjoint on $D_{m|2n}$, with $D=C_0^\infty(\mR^m\backslash \{0\})$, the $C^\infty$ functions with compact support away from the origin, if 
\[
V(u)+\frac{(M-1)(M-3)}{8}\frac{1}{u} \ge \frac{3}{8u}
\]
for $0\le u\le u_0$ for some $u_0>0$.
\end{theorem}

\begin{proof}
The Hilbert space $\bold{L}_2(\mR^{m|2n})$ can be decomposed as 
\[
\bold{L}_2(\mR^{m|2n})=\oplus_{k=0}^{\infty}\bold{L}_2(\mR^{m|2n})_k
\]
with $\bold{L}_2(\mR^{m|2n})_k$ generated by the functions $\phi_{j,k,l}$ with $j\in\mN$, $l=1,\cdots,\dim \cH_k$. We  investigate $H$ on $D_{m|2n}\cap \bold{L}_2(\mR^{m|2n})$, in particular we define $\bold{L}_2^k$ and $\widetilde{\bold{L}_2^k}$ by 
\begin{eqnarray*}
\bold{L}_2(\mR^{m|2n})_k=\bold{L}_2^k\otimes\cH_k&\mbox{and}&\bold{L}_2(\mR^{m|2n})_k\cap D_{m|2n}=\widetilde{\bold{L}_2^k}\otimes\cH_k.
\end{eqnarray*} 
The inner product on $\bold{L}_2^k$ is given by $\langle f(R^2)|g(R^2)\rangle=\int_0^\infty r^{M+2k-1}f(r^2)\overline{g(r^2)}dr$. Using formula (\ref{laplradharm}) we find that on $\bold{L}_2(\mR^{m|2n})_k\cap D_{m|2n}$ the action of $H$ is given by
\[
H|_{\bold{L}_2(\mR^{m|2n})_k\cap D_{m|2n}}=H^{(k)}\otimes I_{\cH_k},
\]
with $I_{\cH_k}$ the unity and 
\begin{eqnarray*}
H^{(k)}f(R^2)=-2R^2f^{(2)}(R^2)-(2k+M)f^{(1)}(R^2)+V(R^2)f(R^2).
\end{eqnarray*}
 By theorem $VIII.33$ in \cite{MR0751959} we need only to prove that for each $k$, $H^{(k)}$ is essentially self-adjoint on $\widetilde{\bold{L}_2^k}$. To prove this we use the isomorphism of Hilbert spaces between $\bold{L}_2^k$ and $L_2(\mR^+,r^{M+2k-1}dr)$ given by $f(R^2)\to f(r^2)$. This isomorphism transforms $H^{(k)}$ into 
\begin{eqnarray*}
-\frac{1}{2}\frac{d^2}{dr^2}-\left(\frac{M+2k-1}{2r}\right)\frac{d}{dr}+V(r^2),
\end{eqnarray*}
which is essentially self-adjoint if 
\[
V(r^2)+\frac{(M-1)(M-3)}{8}\frac{1}{r^2} \ge \frac{3}{8r^2},
\]
for $r$ near zero, see theorem $X.10$ in \cite{MR0493420}.
\end{proof}

\section{The integrability of an $\mathfrak{sl}_2$ representation}
\setcounter{equation}{0}

\label{integrability}

On $\cS(\mR^m)\otimes\Lambda_{2n}$ we have the representation of $\mathfrak{sl}_2$ given by $iR^2/2$, $i\nabla^2/2$ and $\mE+M/2$, which is skew-symmetric in the Hilbert space $\bold{L}_2(\mR^{m|2n})$ by theorem \ref{hermitisch}. As was done in \cite{MR2352481} for Dunkl harmonic analysis we can prove that this representation is integrable using Nelson's theorem (\cite{MR0107176}). Because the representation is densely defined and skew-symmetric we only need to prove that the Casimir operator is essentially self-adjoint. This Casimir operator is given by (\cite{MR1151617})
\begin{eqnarray*}
\cC=(\mE+M/2)^2-\frac{1}{2}(R^2\nabla^2+\nabla^2R^2).
\end{eqnarray*}

To prove the essential self-adjointness of this operator we use the criterion in lemma \ref{critess} and rewrite the Casimir operator as
\begin{eqnarray*}
\cC&=&(\mE+M/2)^2-\frac{1}{4}\left((R^2+\nabla^2)^2-(R^2-\nabla^2)^2\right)\\
&=&(\mE+M/2)^2-\left(\frac{1}{2}(R^2+\nabla^2)\right)^2+\frac{1}{4}(R^2-\nabla^2)^2\\
&=&\frac{1}{2}\left\{\,\mE+M/2+\frac{1}{2}(R^2+\nabla^2)\,\,,\,\,\mE+M/2-\frac{1}{2}(R^2+\nabla^2)\,\right\}+\frac{1}{4}(R^2-\nabla^2)^2.
\end{eqnarray*}
Combining this with lemma \ref{eigCH} yields $\cC\phi_{j,k,l}=\lambda_{j,k,l}\phi_{j,k,l}$ for some constants $\lambda_{j,k,l}$, with $\phi_{j,k,l}$ the complete orthonormal set in $\bold{L}_2(\mR^{m|2n})$ of Hermite functions. This can also be found by calculating the Casimir operator using equations \eqref{commRD} and \eqref{LB}
\[
\cC=\frac{M}{2}\left(\frac{M}{2}-2\right)-\Delta_{LB}.
\]
Equations \eqref{LBH} and \eqref{LBR} then imply
\[
\cC\phi_{j,k,l}=(k+\frac{M}{2}-2)(k+\frac{M}{2})\phi_{j,k,l}.
\]

So we find 
\begin{theorem}
The representation of $\mathfrak{sl}_2$ on $\cS(\mR^m)\otimes \Lambda_{2n}$ given by $iR^2/2$, $i\nabla^2/2$ and $\mE+M/2$  exponentiates to define a unique unitary representation of $\widetilde{SL(2,\mR)}$, the universal covering of $SL(2,\mR)$, on $\bold{L}_2(\mR^{m|2n})$.
\end{theorem}

As a consequence of this and by the properties of the super Fourier transform acting on the Hermite functions (\cite{DBS3, DBS9}), we find that the super Fourier transform (\ref{Four}) can be written as 
\begin{equation}
\label{FT2}
\cF^\pm_{m|2n}=\exp\left(\mp\frac{i\pi M}{4}\right)\,\exp\left(\pm\frac{i\pi}{4}(R^2-\nabla^2)\right),
\end{equation}
which was formally done in \cite{DBS9}.

\section{The uncertainty principle for the super Fourier transform}
\label{Heissection}

\setcounter{equation}{0}
We formulate the Heisenberg uncertainty principle in superspace by means of a Heisenberg inequality for the Fourier transform on $\mR^{m|2n}$. This is a generalization of the bosonic case, see e.g. corollary $2.8$ in \cite{MR1448337}.
\begin{theorem}(Heisenberg inequality)\\
\label{Heis}
For all $f(\bold{x})\in\bold{L}_2(\mR^{m|2n})$ the super Fourier transform satisfies
\begin{eqnarray*}
||Rf(\bold{x})||_0\,||R\cF^{\pm}_{m|2n}(f)(\bold{x})||_0&\ge& \frac{M}{2}||f||_0^2.
\end{eqnarray*}
\end{theorem}
\begin{proof}
First we prove the weaker inequality
\begin{eqnarray}
\label{hulpineq}
||Rf(\bold{x})||_0^2+||R\cF^{\pm}_{m|2n}(f)(\bold{x})||_0^2&\ge& M||f||_0^2.
\end{eqnarray}
The left-hand side of inequality \eqref{hulpineq} is calculated using lemma \ref{Rsymm}, equation \eqref{FourR} and theorem \ref{Pars},
\begin{eqnarray*}
& &\langle Rf(\bold{x})|Rf(\bold{x})\rangle_2+\langle R\cF^{\pm}_{m|2n}(f)(\bold{x})|R\cF^{\pm}_{m|2n}(f)(\bold{x})\rangle_2\\
&=&\langle R^2f(\bold{x})|f(\bold{x})\rangle_2-\langle \cF^{\pm}_{m|2n}(\nabla^2f)(\bold{x})|\cF^{\pm}_{m|2n}(f)(\bold{x})\rangle_2\\
&=&2\langle Hf(\bold{x})|f(\bold{x})\rangle_2=2\langle f| H|f\rangle_2
\end{eqnarray*}
with $H=\frac{1}{2}(R^2-\nabla^2)$ the hermitian hamiltonian of the harmonic oscillator. The inequality \eqref{hulpineq} then follows from the spectrum of the harmonic oscillator, see lemma \ref{eigCH}. The inequality in the theorem can then be calculated in the classical way. First we define $f_{(c)}(\bold{x})$ for $f(\bold{x})\in\bold{L}_2(\mR^{m|2n})$ and $c\in\mR^+$. If $f(\bold{x})=\lim_{s\to\infty}f_s(\bold{x})$ for $f_s\in\cS(\mR^m)_{m|2n}$ then the sequence $\{f_s(c\bold{x})|s\in\mN\}$ is a Banach sequence in $\cS(\mR^m)_{m|2n}$ with respect to the $\bold{L}_2(\mR^{m|2n})$-topology. This can be seen from equation \eqref{superint} which implies
\begin{eqnarray*}
\langle g(c\bold{x})|h(c\bold{x})\rangle_2&=&c^{-M}\langle g(\bold{x})|h(\bold{x})\rangle_2
\end{eqnarray*}
for $g$ and $h$ in $\cS(\mR^{m|2n})$. The function $f_{(c)}$ is then defined as the limit of the Cauchy sequence. This definition implies $||f_{(c)}||_0=c^{-M/2}||f||_0$. A short calculation shows that $\cF^\pm_{m|2n}(f_{(c)})=c^{-M}(\cF^\pm_{m|2n}(f))_{(1/c)}$.

Inequality \eqref{hulpineq} for $f_{(c)}(\bold{x})$ then implies
\begin{eqnarray*}
c^{-M-2}||Rf(\bold{x})||_0^2+c^{2-M}||R\cF^{\pm}_{m|2n}(f)(\bold{x})||_0^2&\ge& Mc^{-M}||f||_0^2.
\end{eqnarray*}
The theorem is then proven by taking $c=||Rf(\bold{x})||_0/||R\cF^{\pm}_{m|2n}(f)(\bold{x})||_0$.
\end{proof}

\begin{corollary}
The inequality in theorem \ref{Heis} is an equality if and only if $f\in \bold{L}_2(\mR^{m|2n})$ is of the form
\begin{eqnarray*}
f(\bold{x})&=&\lambda\exp(-\mu R^2)
\end{eqnarray*}
with $\lambda\in\mC$ and $\mu\in\mR^+$.
\end{corollary}
\begin{proof}
From the proof of theorem \ref{Heis} it is clear that inequality \eqref{hulpineq} is an equality if and only if $f(\bold{x})=\lambda \exp(-R^2/2)$ for a certain $\lambda\in\mC$. Now assume that for $f$ the inequality in theorem \ref{Heis} is an equality. This implies that for $c_0=||Rf(\bold{x})||_0/||R\cF^{\pm}_{m|2n}(f)(\bold{x})||_0$, the equality
\begin{eqnarray*}
c_0^{-2}||Rf(\bold{x})||_0^2+c_0^{2}||R\cF^{\pm}_{m|2n}(f)(\bold{x})||_0^2&=& M||f||_0^2,
\end{eqnarray*}
holds. This implies that for $f_{(c_0)}$ as defined in the proof of theorem \ref{Heis}, equation \eqref{hulpineq} is an equality, so $f_{(c_0)}(\bold{x})=\lambda \exp(-R^2/2)$ or
\begin{eqnarray*}
f(\bold{x})&=&\lambda\exp(-R^2/(2c_0)).
\end{eqnarray*} 
This proves the only if part.

The if part can be proven from the same considerations.
\end{proof}

\section{List of notations}

Some notations for spherical harmonics, Hermite functions and Hilbert spaces used in this paper are listed below.

\begin{tabular}{ll}
$M=m-2n>0$& super-dimension\\
$\ux$& vectorvariable on $\mR^m$\\
$\uxb$& vectorvariable on $\Lambda_{2n}$\\
$\bold{x}$& vectorvariable on $\mR^{m|2n}$\\
$\uy$& vectorvariable on $\mR^M$\\
$\cH_p^b$& spherical harmonics on $\mR^m$ of degree $p$ \\
$\cH_q^f$& spherical harmonics in $\Lambda_{2n}$ of degree $q$ \\
$\cH_k$& spherical harmonics on $\mR^{m|2n}$ of degree $k$\\
$\cH_{M,k}^b$& spherical harmonics on $\mR^M$ of degree $k$\\
$H_p^{b(l)}$& orthonormal basis for $\cH_p^b$ \\
$H_q^{f(t)}$& orthonormal basis for $\cH_f^q$\\
$H_k^{(l)}$& orthonormal basis for $\cH_k$\\
$H_{M,k}^{b(t)}$& orthonormal basis for $\cH_{M,k}^b$\\
$\phi^b_{i,p,l}(\ux)=\frac{1}{\zeta_{i,p}^m}L_i^{\frac{m}{2}+p-1}(r^2)H_p^{b(l)}\exp(-\frac{r^2}{2})$ & Hermite functions on $\mR^{m}$\\
$\phi_{s,q,t}^f(\uxb)=\frac{1}{\zeta_{s,q}^f}L_i^{q-n-1}(\theta^2)H_q^{f(t)}\exp(-\frac{\theta^2}{2})$& Hermite functions in $\Lambda_{2n}$\\
$\phi_{j,k,l}(\bold{x})=\frac{1}{\zeta_{j,k}^M}L_j^{\frac{M}{2}+k-1}(R^2)H_k^{(l)}\exp(-\frac{R^2}{2})$& Hermite functions on $\mR^{m|2n}$\\
$\phi_{j,k,p,q,l,t}=\phi_{j,2k+p+q,r[k,p,q,l,t]}$& Hermite functions on $\mR^{m|2n}$ with basis $\cH$ of lemma \ref{SSin4}\\
$\phi_{j,k,l}(\uy)=\frac{1}{\zeta_{j,k}^M}L_j^{\frac{M}{2}+k-1}(r_{\uy}^2)H_{M,k}^{b(l)}\exp(-\frac{r_{\uy}^2}{2})$& Hermite functions on $\mR^{M}$\\
$L_{2}(\mR^m)$& Hilbert space of square integrable functions on $\mR^m$\\
$L_2(\mR^m)_{m|2n}$& tensor product of $L_2(\mR^m)$ and $\Lambda_{2n}$\\
$\bold{L}_{2}(\mR^{m|2n})$& Hilbert space on $\mR^{m|2n}$ corresponding to $\langle\cdot|\cdot\rangle_2$\\
\end{tabular}


\begin{thebibliography}{10}

\bibitem{MR2352481}
{Ben Sa\"id, S.:}
\newblock {On the integrability of a representation of $\mathfrak{sl}(2,\Bbb R)$.}
\newblock J. Funct. Anal. 250 no. 2, 249--264 (2007)

\bibitem{MR732126}
Berezin, F.A.:
\newblock { Introduction to algebra and analysis with anticommuting
  variables}.
\newblock Moskov. Gos. Univ., Moscow (1983)

\bibitem{MR0238534}
{B\"ohm, A.:}
\newblock {Boulder Lectures on Theoretical Physics, Vol. IXA.}
\newblock Gordon and Breach, New York (1967)

\bibitem{Cauchy}
{Coulembier, K., De~Bie, H., Sommen, F.:}
\newblock { Integration in superspace using distribution theory.}
\newblock  J. Phys. A: Math. Theor. 42, 395206 (2009)

\bibitem{CDBS2}
{Coulembier, K., De~Bie, H., Sommen, F.:}
\newblock Orthogonality of Hermite polynomials in superspace and Mehler type formulae.
\newblock arXiv:1002.1118

\bibitem{CDBS3}
{Coulembier, K., De~Bie, H., Sommen, F.:}
\newblock Orthosymplectically invariant functions in superspace.
\newblock J. Math. Phys. 51, 083504 (2010)

\bibitem{MR1349825}
{Davies, E.:}
\newblock {Spectral theory and differential operators.}
\newblock Cambridge University Press, Cambridge (1995)

\bibitem{DBS9}
De~Bie, H.:
\newblock Fourier transform and related integral transforms in superspace.
\newblock {J. Math. Anal. Appl. 345, 147-164 (2008)}

\bibitem{DBS8}
{De~Bie, H.:}
\newblock Schr\"odinger equation with delta potential in superspace.
\newblock {Phys. Lett. A} {372}, 4350--4352 (2008)

\bibitem{DBE1}
{De~Bie, H., Eelbode, D., Sommen, F.:}
\newblock Spherical harmonics and integration in superspace II.
\newblock {J. Phys. A: Math. Theor. 42}, 245204 (2009)

\bibitem{DBS5}
De~Bie, H., Sommen, F.:
\newblock Spherical harmonics and integration in superspace.
\newblock {J. Phys. A: Math. Theor. 40, 7193-7212 (2007)}

\bibitem{DBS3}
De~Bie, H., Sommen, F.:
\newblock Hermite and Gegenbauer polynomials in superspace using Clifford
  analysis.
\newblock {J. Phys. A: Math. Theor. 40, 10441-10456 (2007)}

\bibitem{MR1292819}
Delbourgo, R., Jarvis, P.D., Warner, R.C.:
\newblock{Schizosymmetry: a new paradigm for superfield expansions.}
\newblock{Modern Phys. Lett. A 9 no. 25, 2305Ð2313 (1994)}

\bibitem{MR1019514}
{Delbourgo, R., Jones, L.M., White, M.:}
\newblock Anharmonic {G}rassmann oscillator.
\newblock {Phys. Rev. D} { 40}, 2716--2719 (1989)

\bibitem{MR1032208}
{Delbourgo, R., Jones, L.M., White, M.:}
\newblock Anharmonic {G}rassmann oscillator {II}.
\newblock {Phys. Rev. D} { 41}, 679--681 (1990)

\bibitem{MR2025382}
{Desrosiers, P., Lapointe, L., Mathieu, P.:}
\newblock Generalized {H}ermite polynomials in superspace as eigenfunctions of
  the supersymmetric rational {CMS} model.
\newblock {Nuclear Phys. B} {674} 3, 615--633 (2003)

\bibitem{MR967935}
{Dunne, G.V., Halliday, I.G.:}
\newblock Negative-dimensional oscillators.
\newblock {Nuclear Phys. B} { 308} 2-3, 589--618 (1988)

\bibitem{MR830398}
Finkelstein, R., Villasante, M.:
\newblock Grassmann oscillator.
\newblock {Phys. Rev. D} { 33}, 1666--1673 (1986)

\bibitem{MR1448337}
Folland, G., Sitaram, A.:
\newblock{The uncertainty principle: a mathematical survey.}
\newblock{J. Fourier Anal. Appl. no. 3, 207--238 (1997)}

\bibitem{MR1151617}
{Howe, R., Tan, E.C.:}
\newblock {Nonabelian harmonic analysis}; Applications of ${\rm SL}(2,{\mR})$.
\newblock Universitext. Springer-Verlag, New York (1992)

\bibitem{MR0107176}
Nelson, E.:
\newblock{ Analytic vectors.}
\newblock{Ann. of Math. (2) 70, 572--615 (1959)}

 \bibitem{MR0751959}
{Reed, M., Simon, B.}:
\newblock {Methods of modern mathematical physics I;} Functional analysis.
\newblock Academic Press, New York-London (1972)

\bibitem{MR0493420}
{Reed, M., Simon, B.:}
\newblock {Methods of modern mathematical physics II;} Fourier analysis, self-adjointness.
\newblock Academic Press, New York-London (1975)


\bibitem{MR1796030}
Rudolph, O.:
\newblock{ Super Hilbert spaces.}
\newblock {Comm. Math. Phys. 214 no. 2, 449--467 (2000)}

\bibitem{MR0209834}
{Schwartz, L:}
\newblock {Th\'eorie des distributions.}
\newblock Hermann, Paris (1966)


\bibitem{MR0278067}
{Simon, B.:}
\newblock Distributions and their Hermite expansions.
\newblock {J. Math. Phys. 12, 140--148 (1971)}

\bibitem{MR2395482}
{Zhang, R.B.:}
\newblock {Orthosymplectic Lie superalgebras in superspace analogues of quantum Kepler problems.}
\newblock   Comm. Math. Phys.  280 no. 2, 545--562  (2008)



\end{thebibliography}
\end{document}